\documentclass{amsart}

\usepackage{a4wide}
\usepackage{amscd}
\usepackage{amssymb}
\usepackage{amsthm}
\usepackage{amsmath}
\usepackage{times}
\usepackage{latexsym}
\usepackage{amsfonts}
\usepackage{graphicx}
\usepackage{graphicx, subfigure}

\newtheorem{theorem}{Theorem}
\theoremstyle{plain}

\newtheorem{definition}[theorem]{Definition}

\newtheorem{lemma}[theorem]{Lemma}

\newtheorem{proposition}[theorem]{Proposition}

\numberwithin{equation}{section}
\numberwithin{theorem}{section}

\newcommand{\R}{\ensuremath{\mathbb{R}}}
\newcommand{\N}{\ensuremath{\mathbb{N}}}
\newcommand{\E}{\ensuremath{\mathbb{E}}}
\newcommand{\<}{\langle}
\renewcommand{\>}{\rangle}

\def\e{{\mathrm{e}}}
\allowdisplaybreaks

\title{Derivatives pricing in energy markets: an infinite dimensional approach}
\author[Benth]{Fred Espen Benth}
\address[Fred Espen Benth]{\\
Department of Mathematics \\
University of Oslo\\
P.O. Box 1053, Blindern\\
N--0316 Oslo, Norway \\
and \\
Centre of Advanced Study \\
Drammensveien 78 \\
N-0271 Oslo, Norway}
\email[]{fredb\@@math.uio.no}
\urladdr{http://folk.uio.no/fredb/}
\author[Kr\"uhner]{Paul Kr\"uhner}
\address[Paul Kr\"uhner]{\\
Department of Mathematics \\
University of Oslo\\
P.O. Box 1053, Blindern\\
N--0316 Oslo, Norway}
\email[]{paulkru\@@math.uio.no}

\date{\today}
\thanks{Financial support from "Managing Weather Risk in Energy Markets (MAWREM)", funded by the Norwegian 
Research Council, is gratefully acknowledged.}
\begin{document}
\maketitle
\begin{abstract}
Based on forward curves modelled as Hilbert-space valued processes, we analyse the pricing of various options
relevant in energy markets. In particular, we connect empirical evidence about energy forward prices known from the literature
to propose stochastic models. Forward prices can be represented as linear functions on a Hilbert space,
and options can thus be viewed as derivatives on the whole curve. The value of these options are computed under
various specifications, in addition to their deltas. In a second part, cross-commodity models are investigated, leading to
a study of square integrable random variables with values in a "two-dimensional" Hilbert space. We analyse the covariance 
operator and representations of such variables, as well as presenting applications to pricing of spread and energy quanto options.  
\end{abstract} 

\section{Introduction}

In energy markets like NYMEX, CME, EEX and NordPool there is a large trade in forwards and 
futures contracts. Forwards and futures on power and gas are delivering the underlying commodity over a period of time 
rather than at a fixed delivery time, as is the case for oil, say. Related markets, like shipping and weather, also trade in
futures and forwards settled on an index measured over a time period. We refer to Burger,
Graeber and Schindlmayr \cite{BGS}, Eydeland and Wolynieck~\cite{EW} and Geman~\cite{Geman} 
for a presentation and discussion of different energy markets and the traded derivatives contracts. For a more technical
analysis on modelling aspects of energy prices, we refer to Benth, \v{S}altyt\.e Benth and Koekebakker~\cite{BSBK-energy}.

Typically, many of the energy markets trade in European call and put options written on the forward and futures contracts,
including for example the power exchanges EEX in Germany and NordPool in the Nordic area. At NYMEX, one finds options
on the spread between futures on different refined oil blends. Other cross-commodity derivatives include options on the spread
between power and fuels (dark and spark spreads, say, see Eydeland and Wolyniec~\cite{EW}), or quanto options which are settled on the product between a power price and a weather index (see Benth, Lange and Myklebust~\cite{BLM}).  

In this paper we analyse the pricing of options in the framework of forward curves modelled as Hilbert-space valued stochastic processes. Empirical studies reveal that energy forwards show a high degree of idiosyncratic risk across
maturities. For example, a principal component analysis of the NordPool power forward and futures market by Benth, \v{S}altyt\.e Benth and Koekebakker~\cite{BSBK-energy} reveal that more than ten factors are needed to explain 95\% of the
volatility (this confirms earlier studies of the same market by Frestad~\cite{Frestad} and Koekebakker and Ollmar~\cite{KO}). 
Using methods from spatial statistics (see Frestad~\cite{Frestad}, Frestad, Benth and Koekebakker~\cite{FBK},
and Andresen, Koekebakker and Westgaard~\cite{AKW}), studies of NordPool forward and
futures prices show a clear correlation structure across times to maturity. These empirical studies point towards
the need for modelling the time dynamics of the forward curve by means of a Hilbert-space valued process. 
Moreover, the above-mentioned studies also highlight the leptokurtic behaviour of price returns, motivating the introduction
of infinite dimensional L\'evy processes as the noise in the forward dynamics. 

This paper develops the analysis of forward curves by Benth and Kr\"uhner~\cite{BK-HJM} towards a theory for pricing 
options in energy markets. In particular, the present paper contributes in two different, but related, directions. Firstly, we
provide a rather detailed study of the pricing of typical European options traded in various energy markets. Secondly,
we lay the foundation for a modelling of cross-commodity forward and futures markets in an infinite dimensional 
framework. 

A European option of a forward contract can, in our context, be viewed as an option on the forward curve.
The payoff of the option will be represented as a linear functional acting on the curve, followed by a non-linear payoff 
function.  We provide a detailed analysis on how to view forward and futures contracts as linear functionals
on the forward curve, set in a Hilbert space of absolutely continuous function on $\R_+$ . We present the explicit 
functionals based on various typical contracts traded in power and weather (temperature) markets. Using a
representation theorem from Benth and Kr\"uhner~\cite{BK-HJM}, one can derive a real-valued stochastic process for the forward contract underlying the option, which in some special cases can be further computed to provided simple
expressions for the option price. For example, for arithmetic (linear) forward curve models we can find 
expression of the option price, either analytical in the Gaussian case, or computable via fast Fourier transform in
the more general L\'evy case. The prices will depend on the realized volatility of the infinite dimensional 
forward curve dynamics, which involves some linear functionals and their duals. In particular, we need to have
available the dual of the shift operator and some integral operators, which we derive explicitly in our chosen
Hilbert space. 

Also, we derive the delta of these options. The delta of the option will be defined as the 
derivative of the price with respect to the initial forward curve. Interestingly, the delta will provide information on how sensitive 
the price is towards inaccuracies on the initial forward curve. As we need to construct this curve from discretely observed data,
the delta provides valuable information on the robustness of the option price towards miss-specification in the forward curve. We also show that the option price is Lipschitz continuous as a function of the initial forward 
curve as long as the payoff function is Lipschitz. In this part of our paper, we also discuss options written on the spread between two forward contracts on the same commodity but with different delivery periods. This spread can effectively be represented as 
the difference of two linear functionals on the forward curve extracting two
different pieces of this curve. With such options, the covariance structure along the forward curve becomes
an important ingredient in the pricing.  

In the second part of the paper we turn the focus to modelling and pricing in cross-commodity energy markets. 
Typically, one is interested in modelling the joint forward dynamics in two energy markets, for example in two
connected power markets or the markets for gas and power. Alternatively, one may be interested in modelling the 
joint forward dynamics between temperature contracts and power. 
We express a bivariate forward price dynamics through a stochastic process with values in a "two-dimensional" Hilbert space. More specifically, we assume that the process is the mild solution of two Musiela stochastic partial differential equations, each
taking values in a Hilbert space of absolutely continuous functions on $\R_+$, where the dynamics is driven by two
dependent Hilbert-space valued Wiener processes. Furthermore, we allow for functional dependency in the volatility specifications
of the two stochastic partial differential equations. The crucial point in our analysis is the covariance operator for the
"bivariate" Hilbert-space valued Wiener process. We show that the covariance operator can be expressed as a $2\times 2$ matrix
of operators, where we find the respective marginal covariance operators on the diagonal and an operator describing the 
covariance between the two Wiener processes on the off-diagonal, analogous to the situation of a bivariate Gaussian 
random variable on $\R^2$.  We derive a decomposition of two square-integrable Hilbert space valued random variables in
terms of a common factor and an independent random variable. This "linear regression" decomposition is expressed in terms of
an operator which resembles the correlation. 

Our theoretical considerations are applied to the pricing of spread options (see Carmona and Durrleman~\cite{CD} for an 
extensive account on the zoology of spread options in energy and commodity markets). 
Another interesting class of derivatives is the 
so-called energy quanto options, which offers the holder a payoff depending on price and volume. The volume component
is measured in terms of some appropriate temperature index, which means that the energy quanto option can be viewed
as an option written on the forward prices of energy and temperatures. We remark that there is a weather market at the
Chicago Mercantile Exchange trading in temperature futures. 
 
Our infinite-dimensional approach to forward price modelling in energy markets builds on the extensive theory in 
fixed-income markets. We refer to Filipovic~\cite{filipovic} and Carmona and Tehranchi~\cite{CT} for an analysis of
forward rates modelled as infinite-dimensional stochastic processes. In Benth and Kr\"uhner~\cite{BK-HJM} a particular 
Hilbert space proposed by Filipovic~\cite{filipovic} to realize forward curves plays a central role. 
Audet {\it et al.}~\cite{audet} is, 
to the best of our knowledge, the first to model power forward prices using infinite dimensional processes. 
Exponential and arithmetic energy forward curve models are analysed in Barth and Benth~\cite{BB} with
an emphasis on introducing numerical schemes to simulate the dynamics. Another path is taken in
Benth and Lempa~\cite{BL}, where optimal portfolio selection in commodity forward markets is studied. 
Barndorff-Nielsen, Benth and Veraart~\cite{BNBV-forward} propose to use ambit fields, a class of spatio-temporal random fields, 
as an alternative modelling approach to the dynamic specification of forward curves used in the present paper.
In a recent paper, Barndorff-Nielsen, Benth and 
Veraart~\cite{BNBV-cross} has extended the ambit field idea to cross-commodity market modelling and the pricing of spread options. We 
remark that there is a close relationship between ambit fields and stochastic partial differential equations 
(see Barndorff-Nielsen, Benth and Veraart~\cite{BNBV-first}).     

We present our results as follows: in Section 2 we express energy forward and futures delivering over a settlement period as 
linear operators on a Hilbert space of functions. European options on energy futures are analysed in Section 3, while 
we consider cross-commodity futures price modelling and option pricing in Section 4. 

\subsection{Some notation}
As a final note in this Introduction, we let throughout this paper $(\Omega,\mathcal{F},\mathcal{F}_t,Q)$ be a 
filtered probability space, where $Q$ denotes the risk-neutral probability. We are working directly under risk-neutrality 
as we have pricing of financial derivatives in mind. Furthermore, we use the notation $L(U,V)$ for the space
of bounded linear operators from the Hilbert space $U$ into the Hilbert space $V$. In case $U=V$, we use the
short-hand notation $L(U)$ for $L(U,U)$. Throughout this paper, the Hilbert spaces that we shall use will
all be assumed separable. Finally, $L_{\text{HS}}(U,V)$ denotes the space of Hilbert-Schmidt operators from $U$ to $V$,
and $L_{\text{HS}}(U)=L_{\text{HS}}(U,U)$.

\section{Hilbert-space realization of energy forwards and futures}
\label{sect:hilbert-space-realization}
In this Section we aim at representing the forward and futures prices in energy markets as
an element of a Hilbert space of functions. Motivated from results in Benth and Kr\"uhner~\cite{BK-HJM},
we will see that various relevant futures contracts traded in energy markets, which deliver the
underlying over a period rather than at a fixed time in the future, can be understood as a
bounded operator on a suitable Hilbert space. 

Let us first introduce the Filipovic space (see Filipovic~\cite{filipovic}), which will be the Hilbert space appropriate 
for our considerations.
Let $H_{\alpha}$ be defined
as the space of all absolutely continuous functions $g:\R_+\rightarrow\R$ for which
$$
\int_0^{\infty}\alpha(x) g'(x)^2\,dx<\infty\,,
$$
for a given continuous and increasing weight function $\alpha:\R_+\rightarrow[1,\infty)$ 
with $\alpha(0)=1$. The norm of $H_{\alpha}$ is $\|g\|_{\alpha}^2=\langle g,g\rangle$, for the inner
product
$$
\langle f,g\rangle=f(0)g(0)+\int_0^{\infty}\alpha(x)g'(x)f'(x)\,dx\,.
$$
Here, $f,g\in H_{\alpha}$. We assume that $\int_0^{\infty}\alpha^{-1}(x)\,dx<\infty$.
Remark that the typical choice of weight function is that of an exponential
function; $\alpha(x)=\exp(\widetilde{\alpha}x)$ for a constant 
$\widetilde{\alpha}>0$,  in which case the integrability condition on the inverse of
$\alpha$ is trivially satisfied.  From Filipovic~\cite{filipovic}, we know that $H_{\alpha}$ is a separable 
Hilbert space. As we shall see, one can realize energy forward and futures prices as linear
operators on $H_{\alpha}$, and in fact interpret energy forward and futures prices as stochastic processes
with values in this space.

Let us consider a simple example motivating the appropriateness of the choice of $H_{\alpha}$.
The classical model for the dynamics of energy spot prices is the so-called
Schwartz dynamics (see Schwartz~\cite{S} and Benth, \v{S}altyt\.{e} Benth and Koekebakker~\cite[Ch. 3]{BSBK-energy}
for an extension to the L\'evy case). Here, the spot price $S(t)$ at time $t\geq 0$ is given by
$$
S(t)=\exp(X(t))\,,
$$
for $X(t)$ being an Ornstein-Uhlenbeck process
$$
dX(t)=\rho(\theta-X(t))\,dt+dL(t)\,,
$$
driven by a L\'evy process $L$. We assume that $L(1)$ has exponential moments, $\rho>0,\theta$ are constants, and $\ln S(0)=X(0)=x\in\R$.  
From Benth, \v{S}altyt\.{e} Benth and Koekebakker~\cite[Prop.~4.6]{BSBK-energy}, 
we find that the forward price $f(t,T)$ at time $t\geq 0$, for a contract delivering at time $T\geq t$, is
$$
f(t,T)=\exp\left(\e^{-\rho(T-t)}X(t)+\theta(1-\e^{-\rho(T-t)})+\int_0^{T-t}\phi(\e^{-\rho s})\,ds\right)\,,
$$ 
with $\phi$ is the logarithm of the moment generating function of $L(1)$.
Recall that we model the spot price directly under the pricing measure $Q$. Letting $x=T-t\geq 0$, we find
(by slightly abusing the notation)
$$
f(t,x)=\exp\left(\e^{-\rho x}X(t)+\theta(1-\e^{-\rho x})+\int_0^{x}\phi(\e^{-\rho s})\,ds\right)\,.
$$ 
It is simple to see that $x\mapsto f(t,x)$ is continuously differentiable for every $t$, and
$$
\frac{\partial f}{\partial x}(t,x)=f(t,x)\left(\rho\e^{-\rho x}(\theta-X(t))+\phi(\e^{-\rho x})\right)\,.
$$
Assume that the weight function $\alpha$ is such that 
$$
\alpha(x)\e^{-2\rho x}\in L^1(\R_+)\,,\qquad \alpha(x)\phi^2(\e^{-2\rho x})\in L^1(\R_+)\,.
$$
Then it follows that  $\int_0^{\infty}|\phi(\exp(-\rho s))|\,ds<\infty$ from the Cauchy-Schwartz inequality and the assumption
$\int_0^{\infty}\alpha^{-1}(x)\,dx<\infty$. Hence, $f$ is uniformly bounded in $x$ since
$$
|f(t,x)|\leq \exp\left(X(t)+\theta+\int_0^{\infty}|\phi(\e^{-\alpha s})|\,ds\right)\,.
$$ 
But then,
\begin{align*}
\Vert f(t,\cdot)\Vert_{\alpha}^2&=|\exp(X(t))|^2+\int_0^{\infty}\alpha(x)f^2(t,x)(\rho\e^{-\rho x}(\theta-X(t))+\phi(\e^{-\rho x}))^2\,dx \\
&\leq c\e^{2X(t)}\left(1+\int_0^{\infty}\alpha(x)\e^{-2\rho x}\,dx+\int_0^{\infty}\alpha(x)\phi^2(\e^{-\rho x})\,dx\right)\,,
\end{align*}
which shows that $f(t,\cdot)\in H_{\alpha}$. If $L$ is a driftless L\'evy process, the exponential moment 
condition on $L(1)$ yields that $\phi(x)$ has the representation
$$
\phi(x)=\frac12\sigma^2x^2+\int_{\R}\{\e^{xz}-1-xz\}\,\ell(dz)\,,
$$
for a constant $\sigma\geq 0$ and L\'evy measure $\ell(dz)$. But by the monotone convergence theorem and
L'Hopital's rule we find that
$$
\lim_{x\searrow 0}\frac1{x^2}\int_{\R}\{\e^{xz}-1-xz\}\,\ell(dz)=\frac12\int_{\R}z^2\,\ell(dz)\,,
$$ 
and therefore $\phi(x)\sim x^2$ when $x$ is small. Thus, a sufficient condition for $f(t,\cdot)\in H_{\alpha}$ is
$\alpha(x)\exp(-2\rho x)\in L^1(\mathbb R_+,\mathbb R)$.

We now move our attention to the main theme of this Section, namely the realization in $H_{\alpha}$ of
general energy forward and futures contracts with a delivery period. Suppose $F(t,T_1,T_2)$ is the 
swap price at time $t$ of a contract on energy delivering
over the time interval $[T_1,T_2]$, where $0\leq t\leq T_1<T_2$. Then one can express 
(see Benth, \v{S}altyt\.e Benth and Koekebakker~\cite{BSBK-energy}, Prop.~4.1) this price as
\begin{equation}
\label{swap-price-eq}
F(t,T_1,T_2)=\int_{T_1}^{T_2}\widetilde{w}(T;T_1,T_2)f(t,T)\,dT
\end{equation}
where $f(t,T)$, $t\leq T$ is the forward price for a contract "delivering energy" at the fixed time time $T$
and $\widetilde{w}(T;T_1,T_2)$ is a deterministic weight function. We will later
make precise assumptions on $\widetilde{w}$, but for now we implicitly assume that the 
integral in \eqref{swap-price-eq} makes sense.  
For example, at the NordPool and EEX power exchanges, swap contracts
deliver electricity over specific weeks, months, quarters and even years, and are of 
either forward or futures style. The delivery is financial, meaning
that the seller of a contract receives the accumulated spot price of power over the specified period
of delivery (forward style) or the interest-rate discounted accumulated spot price (futures style). 
I.e., for these power swap contracts we have the weight function
\begin{equation}
\label{forward-weight}
\widetilde{w}(T;T_1,T_2)=\frac1{T_2-T_1}
\end{equation}
for the forward-style contracts and
\begin{equation}
\label{futures-weight}
\widetilde{w}(T;T_1,T_2)=\frac{\e^{-rT}}{\int_{T_1}^{T_2}\e^{-rs}\,ds}
\end{equation}
for the futures-style. Here, $r>0$ is the risk-free interest rate which we suppose to be 
constant. The reason for the averaging is the market convention of denominating forward
and futures (swap) prices in terms of MWh (Mega Watt hours). 
In the gas market on NYMEX, say, gas is delivered physically at a location (Henry Hub in the 
case of NYMEX) over a given delivery period like month or quarter. We will therefore have the same 
expression \eqref{swap-price-eq} for the gas swap prices as in the case of power swaps. 

Futures on temperature indices like HDD, CDD and 
CAT\footnote{HDD is short-hand for heating-degree days, CDD for cooling-degree days and CAT for cumulative
average temperature. } deliver the money-equivalent from the aggregated
index value over a specified period. Hence, the futures price can be expressed as
$$
F(t,T_1,T_2)=\int_{T_1}^{T_2}f(t,T)\,dT\,,
$$
where $f(t,T)$ is the futures price of a contract that "delivers" the corresponding temperature index at
the fixed delivery time $T\geq t$. I.e., temperature futures can be expressed by \eqref{swap-price-eq} with
\begin{equation}
\label{temp-futures-weight}
\widetilde{w}(T;T_1,T_2)=1\,,
\end{equation}
as the weight function. 
We refer to Benth and \v{S}altyt\.e Benth~\cite{BSB-weather}  for a discussion on weather futures
as well as the definition of various temperature indices. Here one may also find a discussion of
the more recent wind futures, which
can be expressed as the temperature futures except for a different index interpretation of $f$.  

We aim at a so-called Musiela representation of $F(t,T_1,T_2)$ in \eqref{swap-price-eq}. Define  
$x:=T_1-t$, being the time until start of delivery of the swap, and $\ell=T_2-T_1>0$ the length
of delivery of the swap. With the notation $g(t,y):=f(t,t+y)$, one easily derives
\begin{equation}
G^w_{\ell}(t,x):=F(t,t+x,t+x+\ell)=\int_{x}^{x+\ell}w_{\ell}(t,x,y)g(t,y)\,dy\,,
\end{equation}
for the weight function $w_{\ell}(t,x,y)$ defined by
\begin{equation}
w_{\ell}(t,x,y):=\widetilde{w}(t+y;t+x,t+x+\ell)\,,
\end{equation}
where $y\in [x,x+\ell]$, $x\geq 0$ and $t\geq 0$. 
Referring to the different cases of the weight function $\widetilde{w}$, we find that 
$w_{\ell}(t,x,y)=1$ for a temperature (wind) contract  (with $\widetilde{w}$ is as in 
\eqref{temp-futures-weight}) and $w_{\ell}(t,x,y)=1/\ell$ for the forward-style power (gas) swap (using
$\widetilde{w}$ as in \eqref{forward-weight}). Slightly more interesting, is the future-style power swaps, yielding
\begin{equation}
\label{weight:fut-pow-swap}
w_{\ell}(t,x,y)=\frac{r}{1-\e^{-r\ell}}\e^{-r(y-x)}\,.
\end{equation}
Here, we used \eqref{futures-weight}. Note that all these cases result in a weight function $w_{\ell}$ which is 
independent of time. Furthermore, the only case which is depending on $x$ and $y$ is given in
\eqref{weight:fut-pow-swap}, which becomes in fact stationary in the sense that 
$w_{\ell}$ depends on $y-x$. We shall for simplicity restrict to the situation for which 
$w_{\ell}$ is time-independent and stationary. By slightly abusing
notation, we consider weight functions $w_{\ell}:\R_+\rightarrow\R_+$, such that
\begin{equation}
\label{eq:def-swap-musiela}
G^w_{\ell}(t,x)=\int_{x}^{x+\ell}w_{\ell}(y-x)g(t,y)\,dy\,.
\end{equation}
Based on the different cases above, we assume that the weight function $u\mapsto w_{\ell}(u)$ is positive, 
bounded and measurable.  

Following Benth and Kr\"uhner~\cite[Sect.~4]{BK-HJM}, we can represent 
$G^w_{\ell}$ as a linear operator on $g$ after performing a simple integration-by-parts, that is,
$$
G^w_{\ell}(t)=\mathcal{D}^w_{\ell}(g(t))
$$
where, for a generic function $g\in H_{\alpha}$,
\begin{equation}
\label{def-D_ell-op}
\mathcal{D}^w_{\ell}(g)=W_{\ell}(\ell)\text{Id}(g)+\mathcal{I}^w_{\ell}(g)\,.
\end{equation}
Here, $\text{Id}$ is the identity operator and the function $u\mapsto W_{\ell}(u)$, $u\geq 0$ is defined as
\begin{equation}
\label{def-W_ell-funct}
W_{\ell}(u)=\int_0^uw_{\ell}(v)\,dv\,.
\end{equation}
As $w_{\ell}$ is a measurable and bounded function, $W_{\ell}$ is well-defined for
every $u\geq 0$. Note that the limit of $W_{\ell}(u)$ does not necessarily exist when 
$u\rightarrow\infty$. For example, $W_{\ell}$ tends to infinity with $u$ for $w_{\ell}=1/\ell$ or
$w_{\ell}(u)=1$. However, when $w_{\ell}$ is as in \eqref{weight:fut-pow-swap} the limit
of $W_{\ell}$ exists. Since $w_{\ell}$ is positive, the function $u\mapsto W_{\ell}(u)$ is 
increasing. Hence, $W_{\ell}(\ell)>0$, and the first term of $\mathcal{D}_{\ell}^w$ in
\eqref{def-D_ell-op} is simply the indicator operator on $H_{\alpha}$ scaled by the 
positive number $W_{\ell}(\ell)$. Furthermore,
 $\mathcal{I}^w_{\ell}$ in \eqref{def-D_ell-op} is an integral operator
\begin{equation}
\label{def-I_ell-op}
\mathcal{I}^w_{\ell}(g)=\int_{0}^{\infty}q^w_{\ell}(\cdot,y)g'(y)\,dy\,,
\end{equation}
with kernel
\begin{equation}
\label{def-g_ell-kernel}
q^w_{\ell}(x,y)=(W_{\ell}(\ell)-W_{\ell}(y-x))\mathrm{1}_{[0,\ell]}(y-x)\,.
\end{equation}
Before we show that $\mathcal{I}_{\ell}^w$ is a bounded operator on
$H_{\alpha}$, we look at a special case:

Consider a simple forward-style power swap, i.e., $w_{\ell}(u)=1/\ell$. We get
$W_{\ell}(u)=u/\ell$, and therefore $W_{\ell}(\ell)=1$ yielding that first term in 
\eqref{def-D_ell-op} is simply the identity operator on $H_{\alpha}$. The integral operator
$\mathcal{I}^w_{\ell}$ has the kernel 
$$
q^w_{\ell}(x,y)=\frac1{\ell}(x+\ell-y)\mathrm{1}_{[x,x+\ell]}(y)\,.
$$
This example is analysed in Benth and Kr\"uhner~\cite[Sect.~4]{BK-HJM}. They show that the integral
operator $I_{\ell}^w$ in this case is a bounded linear operator on $H_{\alpha}$, 
implying that $t\mapsto G_{\ell}^w(t)$ is  is a stochastic process with values in $H_{\alpha}$ as 
long as $t\mapsto g(t)$ is an $H_{\alpha}$-valued process. 
It turns out that the boundedness property of the integral operator $\mathcal{I}^w_{\ell}$ holds 
also for our class of more general weight functions. This is shown in
the next Proposition: 
\begin{proposition}
\label{prop:cont-I}
Under the assumption that $u\mapsto w_{\ell}(u)$ for $u\in\R_+$ is positive, bounded and measurable, it holds that $\mathcal{I}_{\ell}^w$ is a bounded linear operator on $H_{\alpha}$. 
\end{proposition}
\begin{proof}
Obviously, $q_{\ell}^w(x,y)$ is measurable on $\R_+^2$. Moreover, it is bounded since for $y\in[x,x+\ell]$
$$
0\leq W_{\ell}(\ell)-W_{\ell}(y-x)=\int_{y-x}^{\ell}w_{\ell}(u)\,du\leq c\ell\,,
$$
where $c$ is the constant majorizing $w_{\ell}$. Hence, $0\leq q^w_{\ell}(x,y)\leq c\ell$. It follows that
$$
\int_0^{\infty}\alpha^{-1}(y)(q_{\ell}^w(x,y))^2\,dy\leq c^2\ell^2\int_0^{\infty}\alpha^{-1}(y)\,dy<\infty
$$
and part 1 of Cor.~4.5 in Benth and Kr\"uhner \cite{BK-HJM} holds. This implies that the integral 
operator $\mathcal{I}^w_{\ell}$ is defined for all $g\in H_{\alpha}$. We continue to
demonstrate that part 2 of the same Corollary also holds. 

As short-hand notation, let for a given $g\in H_{\alpha}$,
$$
\xi(x):=\int_0^{\infty}q_{\ell}^w(x,y)g'(y)\,dy=\int_x^{x+\ell}(W_{\ell}(\ell)-W_{\ell}(y-x))g'(y)\,dy\,.
$$
In particular,
$$
\xi(0)=\int_0^{\ell}(W_{\ell}(\ell)-W_{\ell}(y))g'(y)\,dy=\int_0^{\ell}\int_y^{\ell}w_{\ell}(u)\,du\,g'(y)\,dy\,.
$$
Hence, we find
\begin{align*}
\xi^2(0)&=\left(\int_0^{\ell}\int_y^{\ell}w_{\ell}(u)\,du\,g'(y)\,dy\right)^2 \\
&\leq \left(\int_0^{\ell}\int_y^{\ell}w_{\ell}(u)\,du|g'(y)|\,dy\right)^2 \\
&\leq \left(\int_0^{\ell}w_{\ell}(u)\,du\right)^2\left(\int_0^{\ell}|g'(y)|\,dy\right)^2 \\
&=W_{\ell}^2(\ell)\left(\int_0^{\ell}\sqrt{\alpha(y)}|g'(y)|\sqrt{\alpha(y)}^{-1}\,dy\right)^2 \\
&\leq W_{\ell}^2(\ell)\int_0^{\ell}\alpha^{-1}(y)\,dy\int_0^{\ell}\alpha(y)g'(y)^2\,dy \\
&\leq W_{\ell}^2(\ell)\int_0^{\ell}\alpha^{-1}(y)\,dy\|g\|_{\alpha}^2\,,
\end{align*}
where, in the second inequality we used that $w_{\ell}$ is positive and in the third the Cauchy-Schwartz
inequality. Recall that by assumption, $\int_0^{\infty}\alpha^{-1}(y)\,dy<\infty$.
Furthermore, it holds that 
\begin{align*}
\xi'(x)&=\frac{d}{dx}\int_x^{x+\ell}(W_{\ell}(\ell)-W_{\ell}(y-x))g'(y)\,dy \\
&=(W_{\ell}(\ell)-W_{\ell}(\ell))g'(x+\ell)-(W_{\ell}(\ell)-W_{\ell}(0))g'(x) \\
&\qquad+\int_x^{x+\ell}(-W_{\ell}'(y-x))(-1)
g'(y)\,dy \\
&=\int_x^{x+\ell}w_{\ell}(y-x)g'(y)\,dy-W_{\ell}(\ell)g'(x)\,,
\end{align*}
and therefore $\xi$ has a (weak) derivative. By the triangle inequality,
$$
\xi'(x)^2\leq2W_{\ell}(\ell)g'(x)^2+2\left(\int_x^{x+\ell}w_{\ell}(y-x)g'(y)\,dy\right)^2\,.
$$
We consider the second term on the right hand side: By Cauchy-Schwartz' inequality and
boundedness of $w_{\ell}$,
\begin{align*}
\int_0^{\infty}\alpha(x)\left(\int_x^{x+\ell}w_{\ell}(y-x)g'(y)\,dy\right)^2\,dx&\leq 
\int_0^{\infty}\alpha(x)\left(\int_x^{x+\ell}w_{\ell}(y-x)|g'(y)|\,dy\right)^2\,dx \\
&\leq \int_0^{\infty}\alpha(x)\int_x^{x+\ell}w^2_{\ell}(y-x)\,dy\int_x^{x+\ell}g'(y)^2\,dy\,dx \\
&\leq c^2\ell\int_0^{\infty}\int_x^{x+\ell}\alpha(y)g'(y)^2\,dy\,dx \\ 
&\leq c^2\ell^2\|g\|_{\alpha}^2\,,
\end{align*}
after using that $\alpha$ is non-decreasing and Fubini's Theorem. Wrapping up these estimates, we 
majorize the $H_{\alpha}$-norm of $\xi$
\begin{align*}
\|\xi\|_{\alpha}^2&=|\xi^2(0)|+\int_0^{\infty}\alpha(x)\xi'(x)^2\,dx \\
&\leq W_{\ell}^2(\ell)\int_0^{\ell}\alpha^{-1}(y)\,dy\|g\|_{\alpha}^2+2W^2_{\ell}(\ell)\|g\|_{\alpha}^2
+2c^2\ell^2\|g\|_{\alpha}^2 \\
&\leq C\|g\|_{\alpha}^2\,,
\end{align*}
for a positive constant $C$. But then $\xi\in H_{\alpha}$, and we can conclude from from Cor.~4.5 of
Benth and Kr\"uhner~\cite{BK-HJM} that $\mathcal{I}_{\ell}^w$ is a continuous linear operator
on $H_{\alpha}$. The Proposition follows.
\end{proof}
From Prop.~\ref{prop:cont-I} it follows immediately that $\mathcal{D}_{\ell}^w$ in \eqref{def-D_ell-op} is a continuous linear
operator on $H_{\alpha}$, as it is the sum of the scaled identity operator and the integral
operator $\mathcal{I}_{\ell}^w$. Moreover, for $g\in H_{\alpha}$, it holds (by inspection of the proof of
Prop.~\ref{prop:cont-I}) that 
$$
\||\mathcal{D}_{\ell}^w(g)\|_{\alpha}\leq\left\{W_{\ell}(\ell)+\sqrt{W^2_{\ell}(\ell)(2+\int_0^{\ell}\alpha^{-1}(y)\,dy)+2c^2\ell^2}\right\}\,\|g\|_{\alpha}\,,
$$ 
which provides us with an upper bound on the operator norm of $\mathcal{D}_{\ell}^w$. Furthermore,
it follows immediately from Prop.~\ref{prop:cont-I} that we can realize the dynamics of swap price
curves in $H_{\alpha}$. E.g., if $g(t)$ is an $H_{\alpha}$-valued stochastic process, then
$t\mapsto G_{\ell}^w(t)$ will be a stochastic process with values in $H_{\alpha}$ as well. 
                                                                                                                                                                                                                                                                                                                                                                                                                                                                                                                                                                                                                                                                                                                                                                                                                                                                               
\section{European options on energy forwards and futures}
\label{sect:european-option}
At the energy exchanges, plain vanilla call and put options are offered for trade on futures and forward contracts.
For example, at NordPool, one can buy and sell options on the quarterly settled power futures contracts, while 
at CME one can trade in options on weather futures, including HDD/CDD and CAT temperature futures.
NYMEX offer trade in options on gas futures, among a number of other derivatives on energy and commodity
futures (including different blends of oil).

Consider a European option on an energy forward contract delivering over the period $[T_1,T_2]$ and price $F(t,T_1,T_2)$ at time $t$, where the option has exercise time $0\leq \tau\leq T_1$ and payoff $p(F(\tau,T_1,T_2))$ for some function $p:\mathbb{R}\rightarrow\mathbb{R}$. For plain-vanilla call and put options, we have
$p(x)=\max(x-K,0)$ or $p(x)=\max(K-x,0)$, resp., with the strike price denoted $K$.  
We assume in general $p$ to be a measurable function of 
at most linear growth. We recall the representation $F(t,T_1,T_2) = \mathcal D_\ell^w(g(t))(T_1-t)$. The following Proposition provides the link to the infinite dimensional swap prices:
\begin{proposition}
\label{prop:payoff-repr}
Suppose $p$ is of at most linear growth. It holds that
$$
p(F(\tau,T_1,T_2))=\mathcal{P}_{\ell}(T_1-\tau,g(\tau))\,,
$$
for a nonlinear functional $\mathcal{P}^w_{\ell}:\mathbb{R}_+\times H_{\alpha}\rightarrow\mathbb{R}$ defined by
$$
\mathcal{P}^w_{\ell}(x,g)=p\circ\delta_x\circ\mathcal{D}^w_{\ell}(g)\,.
$$ 
Here, $\ell=T_2-T_1$. Moreover, there exists a constant $c_{\ell}>0$ depending on $\ell$ such that,
$$
|\mathcal{P}^w_{\ell}(\cdot,g)|_{\infty}\leq c_{\ell}(1+\|g\|_{\alpha})\,.
$$
\end{proposition}
\begin{proof}
Since we have $F(\tau,T_1,T_2)=G^w_{T_2-T_1}(\tau,T_1-\tau)$, the first claim follows.
From the linear growth of $p$ , 
we find
\begin{align*}
|\mathcal{P}^w_{\ell}(x,g)|&=|p(\mathcal{D}^w_{\ell}(g)(x))|\leq c_1\left(1+|\mathcal{D}^w_{\ell}(g)(x)|\right)\,,
\end{align*}
for a positive constant $c_1$. 
Since $\int_0^{\infty}\alpha^{-1}(y)\,dy<\infty$, we find by Lemma 3.2 in Benth and Kr\"uhner~\cite{BK-HJM},
\begin{align*}
|\mathcal{P}^w_{\ell}(\cdot,g)|_{\infty}&=\sup_{x\in\R_+}|\mathcal{P}_{\ell}(x,g)|
\leq c_2\left(1+\|\mathcal{D}^w_{\ell}(g)\|_{\alpha}\right)\,,
\end{align*}
for a positive constant $c_2>0$. But $\mathcal{D}^w_{\ell}$ is a continuous linear operator on $H_{\alpha}$ 
by Prop.~\ref{prop:cont-I}, and hence so is $\mathcal{D}_{\ell}^w$. The last claim follows and the 
proof is complete. 
\end{proof}

Consider the special case of power forwards, for which we recall that $w_{\ell}(u)=1/\ell$. In this case we observe 
$$
\lim_{\ell\downarrow0}G^w_{\ell}(t,x)=\frac{\partial}{\partial\ell}\int_x^{x+\ell}g(t,y)\,dy
\vert_{\ell=0}=g(t,x)\,.
$$
Hence, we can make sense out of $\mathcal{P}^w_0$ for $w_{\ell}(u)=1/\ell$ as
\begin{equation}
\mathcal{P}_0(x,g)=p\circ\delta_x(g)\,.
\end{equation}
Here, $x\in\R_+$ and $g\in H_{\alpha}$, and we use the simplified notation $\mathcal{P}_0$ instead of
$\mathcal{P}_0^w$ in this particular case. We note that the nonlinear operator $\mathcal{P}_0$ will be the 
payoff from an option on a forward with fixed time to delivery $x$ instead of
a delivery period which lasts $\ell>0$, since it holds
\begin{equation}
p(f(\tau,T))=\mathcal{P}_0(T-\tau,g(\tau))\,,
\end{equation}
for $\tau\leq T$. The markets for oil at NYMEX, for example, trade in forwards and futures with fixed delivery times, and options
on these contracts. It is straightforward from Lemma~3.2 in Benth and Kr\"uhner~\cite{BK-HJM} 
that 
$$
|\mathcal{P}_{0}(\cdot,g)|_{\infty}=\sup_{x\in\R_+}|p(g(x))|\leq c_1(1+\sup_{x\in\R_+}|g(x)|) \leq c_2\left(1+
\Vert g\Vert_{\alpha}\right)\,,
$$
for $g\in H_{\alpha}$ and a payoff function $p$ with at most linear growth.


Suppose now that $g(t)$ is a stochastic process in $H_{\alpha}$ satisfying
\begin{equation}
\label{eq:cond-g-integrable}
\E\left[\Vert g(t)\Vert_{\alpha}\right]<\infty\,,
\end{equation}
for all $t\geq 0$. The price $V(t)$ at time $0\leq t\leq\tau$ of the option with payoff $p(F(\tau;T_1,T_2)$ 
at time $0<\tau\leq T_1$ is given as
\begin{equation}
\label{option-val-formula}
V(t)=\e^{-r(\tau-t)}\E\left[\mathcal{P}^w_{\ell}(T_1-\tau,g(\tau))\,|\,\mathcal{F}_t\right]\,.
\end{equation}
The expectation is well-defined by Prop.~\ref{prop:payoff-repr} for any given $\ell>0$. 
If we select $w_{\ell}(u)=1/\ell$, then the option value in \eqref{option-val-formula} also incorporates contracts
written on fixed-delivery forwards, that is, options with payoff $p(f(\tau,T))$,
\begin{equation}
\label{option-val-formula-x=0}
V(t)=\e^{-r(\tau-t)}\E\left[\mathcal{P}_{0}(T-\tau,g(\tau))\,|\,\mathcal{F}_t\right]\,.
\end{equation} 
This is also well-defined under the assumption \eqref{eq:cond-g-integrable}. 

\subsection{Markovian forward curves}\label{sect:Markovian}
We want to analyse option prices for a class of Markovian forward curve dynamics, where the process 
$g(t)$ is specified as the solution of a (first-order) stochastic partial differential equation. We shall be concerned with
dynamics driven by an infinite-dimensional L\'evy process. 

Before proceeding, let us first introduce some general
notions (see e.g. Peszat and Zabczyk~\cite{peszat.zabczyk.07} for what follows):
A random variable $X$ with values in a separable Hilbert space $H$ 
is {\em square integrable} if $\E( \Vert X\Vert^2) <\infty$. If $X$ is square 
integrable, $\mathcal{Q}\in L(H)$ is called the {\em covariance operator} of $X$ if
$$ 
\E\left( \langle X,u\rangle\langle X,v\rangle \right) = \langle \mathcal{Q}u,v\rangle\,, 
$$
for any $u,v\in H$. Here, $\langle\cdot,\cdot\rangle$ is the inner 
product in $H$ and  $\|\cdot\|$ the associated norm.
The following result can be found in Peszat and Zabczyk~\cite[Thm.~4.44]{peszat.zabczyk.07}, 
and stated here for convenience.
\begin{lemma}
Let $X$ be a square integrable $H$-valued random variable where $H$ is a 
separable Hilbert space. Then there is a unique operator 
$\mathcal{Q}\in L(H)$ such that 
$\mathcal{Q}$ is the covariance operator of $X$. 
Moreover, $\mathcal{Q}$ is a positive semidefinite trace class operator. Consequently, there is an orthonormal basis $(e_n)_{n\in I}$ of $H$ and a sequence $(\lambda_n)_{n\in I}\in l^1(I,\mathbb R_+)$ such that
$$
\mathcal{Q}u = \sum_{n\in\mathbb N} \lambda_n \langle e_n,u\rangle e_n
$$
for any $u\in H$.
\end{lemma}
For a separable Hilbert space $H$, $\mathbb{L}:=\{\mathbb{L}(t)\}_{t\geq 0}$ is an $H$-valued L\'evy process 
if $\mathbb{L}$ has independent and stationary increments, stochastically continuous paths and
$\mathbb{L}(0)=0$. This definition is found in Peszat and Zabczyk~\cite[Ch. 4]{peszat.zabczyk.07},
and can in fact be formulated on a general Banach space. We remark in passing that Thm.~4.44 in 
Peszat and Zabczyk~\cite{peszat.zabczyk.07} is formulated for L\'evy processes.

Let us now move our attention back to modelling the forward rate dynamics, and 
suppose that  $\mathbb{L}$ is a square-integrable $H$-valued L\'evy process with zero mean, and
denote its covariance operator by $\mathcal{Q}$. Furthermore, let 
$\sigma:\R_+\times H_{\alpha}\rightarrow L(H,H_{\alpha})$ be a measurable map, and assume there
exists an increasing function $K:\R_+\rightarrow\R_+$ such that the following Lipschitz continuity and linear growth
holds: for any $f,h\in H_{\alpha}$ and $t\in\R_+$,
\begin{align}
\|\sigma(t,f)-\sigma(t,h)\|_{\text{op}}&\leq K(t)\|f-h\|_{\alpha} \,,  \label{eq:sigma-lipschitz}    \\
\|\sigma(t,f)\|_{\text{op}}&\leq K(t)(1+\|f\|_{\alpha})\,. \label{eq:sigma-lingrowth}
\end{align}
Consider the dynamics of the $H_{\alpha}$-valued stochastic process $\{g(t)\}_{t\geq 0}$ 
defined by the stochastic partial differential equation
\begin{equation}
\label{eq:g-spde}
dg(t)=\partial_xg(t)\,dt+\sigma(t,g(t))\,d\mathbb{L}(t)\,.
\end{equation}
Let $\mathcal{S}_x, x\geq 0$ denote the right-shift operator on $H_{\alpha}$, i.e., $\mathcal{S}_xf=f(x+\cdot)$.
Then $\mathcal{S}_x$ is the $C_0$-semigroup generated by the operator $\partial_x$ 
(see Filipovic~\cite[Thm.~5.1.1]{filipovic}). 
From Lemma~3.5 in Benth and Kr\"uhner~\cite{BK-HJM}, $\mathcal{S}_x$ is quasi-contractive, i.e., there exists
a positive constant $c$ such that $\|\mathcal{S}_x\|_{\text{op}}\leq \exp(ct)$ for $t >0$. 
Hence, referring to Thm.~4.5 in Tappe~\cite{Tappe}, there exists a unique mild solution of \eqref{eq:g-spde} 
for $s\geq t$, that is, a c\`adl\`ag process $g\in H_{\alpha}$ satisfying
\begin{equation}
\label{eq:g-markovian}
g(s)=\mathcal{S}_{s-t}g(t)+\int_t^s\mathcal{S}_{s-u}\sigma(u,g(u))\,d\mathbb{L}(u)\,.
\end{equation}

The shift and the pricing operator for $F(t,T_1,T_2)$ commute, which allows to find the dynamics of $F(\cdot,T_1,T_2)$.
Moreover, this dynamics reveals that $t\mapsto F(t,T_1,T_2)$ is a martingale in our setup, as desired.
\begin{lemma}\label{lemma:forward-dynamics}
 We have $\mathcal S_x\mathcal D_\ell^w = \mathcal D_\ell^w\mathcal S_x$ for any $x\geq 0$. Consequently, we have
  \begin{equation}
\label{eq:F simple}
F(s,T_1,T_2) = \delta_{T_1-t}\mathcal D_\ell^w g(t)+\int_t^s\delta_{T_1-u}\mathcal D_\ell^w\sigma(u,g(u))\,d\mathbb{L}(u)\, 
\end{equation}
for any $0\leq t\leq s$.
\end{lemma}
\begin{proof}
 The first equality follows from a straightforward computation. Applying the mild solution in Equation \eqref{eq:g-markovian} to $F(s,T_1,T_2) = \delta_{T_1-s}\mathcal D_\ell^wg(s)$, the claim follows after using the commutation property.
\end{proof}

Below it will be convenient to know that $\mathcal{S}_x$ is uniformly bounded in the operator norm:  
\begin{lemma}
\label{lem:cont-shift-op}
It holds that 
$\|\mathcal{S}_x\|_{\text{op}}^2\leq2\max(1,\int_0^{\infty}\alpha^{-1}(y)\,dy)$ for $x\geq 0$.
\end{lemma}
\begin{proof}
This follows by a direct calculation: By the fundamental theorem of
calculus, the elementary inequality $2ab\leq a^2+b^2$ and $\alpha$ being non-decreasing, we find for
$f\in H_{\alpha}$
\begin{align*}
\|\mathcal{S}_xf\|^2_{\alpha}&=f^2(x)+\int_0^{\infty}\alpha(y)|f'(x+y)|^2\,dy \\
&=\left(f(0)+\int_0^xf'(y)\,dy\right)^2+\int_x^{\infty}\alpha(y-x)|f'(y)|^2\,dy \\
&\leq 2f^2(0)+2\left(\int_0^x\alpha^{-1/2}(y)\alpha^{1/2}(y)f'(y)\,dy\right)^2+
\int_x^{\infty}\alpha(y)|f'(y)|^2\,dy \,.
\end{align*}
Appealing to the Cauchy-Schwartz inequality we find,
\begin{align*}
\|\mathcal{S}_xf\|^2_{\alpha}&\leq 2f^2(0)+2\int_0^x\alpha^{-1}(y)\,dy\int_0^x\alpha(y)|f'(y)|^2\,dy+
\int_x^{\infty}\alpha(y)|f'(y)|^2\,dy\,.
\end{align*}
Hence, $\|\mathcal{S}_xf\|_{\alpha}^2\leq \max(2,2\int_0^{\infty}\alpha^{-1}(y)\,dy)\|f\|_{\alpha}^2$, and the Lemma follows.
\end{proof}

From \eqref{eq:g-markovian}, the dynamics of $g$ becomes Markovian. This means
in particular that $V(t)$ defined in \eqref{option-val-formula} can be expressed as $V(t)=V(t,g(t))$ 
(with a slight abuse of notation) for 
\begin{equation}
V(t,g)=\e^{-r(\tau-t)}\E\left[\mathcal{P}_{\ell}(g^{t,g}(\tau))\right]\,.
\end{equation}
Here, we have used the notation $g^{t,g}(s)\,s\geq t$ for the for process $g(s), s\geq t$, starting in $g$ at time 
$t$, e.g., $g^{t,g}(t)=g$, $g\in H_{\alpha}$.  

We shall use the continuity of the translation operator as a linear operator on $H_{\alpha}$ to prove Lipschitz continuity
of the functional $g\mapsto V(t,g)$, uniformly in $t\leq\tau$. Recall that $\tau$ is the exercise time
of the option in question.
\begin{proposition}
\label{prop:V-lipschitz}
Assume that the payoff function $p$ is Lipschitz continuous and volatility functional $g\mapsto \sigma(s,g)$ satisfies
the Lipschitz and linear growth conditions in (\ref{eq:sigma-lipschitz},\ref{eq:sigma-lingrowth}). Then there
exists a positive constant $C$ (depending on $\tau$) such that
$$
\sup_{t\leq\tau}|V(t,g)-V(t,\widetilde{g})|\leq C\|g-\widetilde{g}\|_{\alpha}\,,
$$ 
for $g,\widetilde{g}\in H_{\alpha}$.
\end{proposition}  
\begin{proof}
As $p$ is Lipschitz continuous, it follows that $g\mapsto\mathcal{P}_{\ell}(x,g)$ is Lipschitz 
continuous since $\mathcal{P}_{\ell}(x,\cdot)=p\circ\delta_x\circ\mathcal{D}_{\ell}^w$, and $\delta_x,\mathcal{D}_{\ell}^w$ are bounded linear operators. Moreover, the Lipschitz continuity is uniform in $x$, as it follows from Lemma 3.1 in 
Benth and Kr\"uhner~\cite{BK-HJM} that the operator norm of $\delta_x$ satisfies
$$
\|\delta_x\|_{\text{op}}^2=1+\int_0^x\alpha^{-1}(y)\,dy\leq 1+\int_0^{\infty}\alpha^{-1}(y)\,dy<\infty\,.
$$   
Hence, there exists a constant $C_{\mathcal{P}}>0$ such that
$$
|\mathcal{P}_{\ell}(x,g)-\mathcal{P}_{\ell}(x,\widetilde{g})|\leq C_{\mathcal{P}}\|g-\widetilde{g}\|_{\alpha}\,.
$$
Therefore 
$$
|V(t,g)-V(t,\widetilde{g})|\leq C_{\mathcal{P}}\E\left[\|g^{t,g}(\tau)-g^{t,\widetilde{g}}(\tau)\|_{\alpha}\right]\,.
$$
Since
$$
g^{t,g}(\tau)=\mathcal{S}_{\tau-t}g+\int_t^{\tau}\mathcal{S}_{\tau-s}\sigma(s,g^{t,g}(s))\,\mathbb{L}(s)
$$
we have by the triangle inequality and Lemma~\ref{lem:cont-shift-op}
\begin{align*}
\|g^{t,g}(\tau)-g^{t,\widetilde{g}}(\tau)\|_{\alpha}&\leq\|\mathcal{S}_{\tau-t}(g-\widetilde{g})\|_{\alpha}
+\|\int_t^{\tau}\mathcal{S}_{\tau-s}\left(\sigma(s,g^{t,g}(s))-\sigma(s,g^{t,\widetilde{g}}(s))\right)\,d\mathbb{L}(s)\|_{\alpha} \\
&\leq c\|g-\widetilde{g}\|_{\alpha}+\|\int_t^{\tau}\mathcal{S}_{\tau-s}\left(\sigma(s,g^{t,g}(s))-\sigma(s,g^{t,\widetilde{g}}(s))\right)\,d\mathbb{L}(s)\|_{\alpha} \,,
\end{align*}
where the constant $c$ is positive, and in fact given explicitly in Lemma~\ref{lem:cont-shift-op}.
By the It\^o isometry it follows,
\begin{align*}
\E&\left[\|\int_t^{\tau}\mathcal{S}_{\tau-s}\left(\sigma(s,g^{t,g}(s))-\sigma(s,g^{t,\widetilde{g}}(s))\right)\,d\mathbb{L}(s)\|_{\alpha}^2\right] \\
&\qquad\qquad\qquad=\int_t^{\tau}\E\left[\|\mathcal{S}_{\tau-s}\left(\sigma(s,g^{t,g}(s))-\sigma(s,g^{t,\widetilde{g}}(s))\right)\mathcal{Q}^{1/2}\|^2_{L_{\text{HS}}(H,H_{\alpha})}\right]\,ds\,.
\end{align*}
Let now $\mathcal{T}\in L(H,H_{\alpha})$. Then, we have
 \begin{align*}
  \Vert \mathcal S_x\mathcal T\mathcal Q^{1/2}\Vert_{L_{HS}(H,H_\alpha)} &\leq \Vert \mathcal S_x\Vert_{\mathrm{op}} \Vert \mathcal T\Vert_{\mathrm{op}}\Vert \mathcal Q^{1/2}\Vert_{L_{HS}(H)} \\
   &\leq c\Vert \mathcal T\Vert_{\mathrm{op}}\Vert \mathcal Q^{1/2}\Vert_{L_{HS}(H)}.
 \end{align*}
%
%
Letting $\mathcal{T}=\sigma(s,g^{t,g}(s))-\sigma(s,g^{t,\widetilde{g}}(s))$ and $x=\tau-s$, we find
from the Lipschitz continuity of $\sigma$ in \eqref{eq:sigma-lipschitz}
\begin{align*}
\|\mathcal{S}_{\tau-s}&\left(\sigma(s,g^{t,g}(s))-\sigma(s,g^{t,\widetilde{g}}(s))\right)\mathcal{Q}^{1/2}\|^2_{L_{\text{HS}}(H,H_{\alpha})} \\
&\qquad\qquad\leq c^2\|\mathcal{Q}^{1/2}\|^2_{L_{\text{HS}}(H)}
\|\sigma(s,g^{t,g}(s))-\sigma(s,g^{t,\widetilde{g}}(s))\|_{\text{op}}^2\\
&\qquad\qquad\leq c^2 K^2(s)\|\mathcal{Q}^{1/2}\|^2_{L_{\text{HS}}(H)}
\|g^{t,g}(s)-g^{t,\widetilde{g}}(s)\|^2_{\alpha}\,.
\end{align*}
But $K$ is an increasing function in the Lipschitz continuity of $\sigma$, so $K(s)\leq K(\tau)$. Hence,
\begin{align*}
\E&\left[\|\int_t^{\tau}\mathcal{S}_{\tau-s}\left(\sigma(s,g^{t,g}(s))-\sigma(s,g^{t,\widetilde{g}}(s))\right)\,d\mathbb{L}(s)\|_{\alpha}^2\right] \\
&\qquad\qquad\qquad\leq c^2 K^2(\tau)\|\mathcal{Q}^{1/2}\|^2_{L_{\text{HS}}(H)}
\int_t^{\tau}\E\left[\|g^{t,g}(s)-g^{t,\widetilde{g}}(s)\|^2_{\alpha}\right]ds\,.
\end{align*}
If we now apply the elementary inequality $(a+b)^2\leq 2a^2+2b^2$, we derive
\begin{align*}
\E\left[\|g^{t,g}(\tau)-g^{t,\widetilde{g}}(\tau)\|_{\alpha}^2\right]&\leq 2c^2
\|g-\widetilde{g}\|_{\alpha}^2 \\
&\qquad\qquad+2c^2\|\mathcal{Q}^{1/2}\|^2_{L_{\text{HS}}(H)}K^2(\tau)
\int_t^{\tau}\E\left[\|g^{t,g}(s)-g^{t,\widetilde{g}}(s)\|_{\alpha}^2\right]\,ds\,.
\end{align*}
Gr\"onwall's inequality then yields,
$$
\E\left[\|g^{t,g}(\tau)-g^{t,\widetilde{g}}(\tau)\|_{\alpha}^2\right]\leq2c\e^{2c
\|\mathcal{Q}^{1/2}\|_{L_{\text{HS}}(H)}^2K^2(\tau)(\tau-t)}\|g-\widetilde{g}\|_{\alpha}^2\,.
$$
From Jensen's inequality we thus derive
$$
|V(t,g)-V(t,\widetilde{g})|\leq C_{\mathcal{P}}\sqrt{2c}
\e^{(2cK^2(\tau)\|\mathcal{Q}^{1/2}\|_{L_{\text{HS}}(H)}^2\tau}\|g-\widetilde{g}\|_{\alpha}\,,
$$
and the result follows.
\end{proof}
The Proposition shows that the option price is uniformly Lipschitz continuous in the initial 
forward curve as long as we consider Lipschitz continuous payoff functions and 
volatility operators $\sigma$. 
We remark that put and call options have Lipschitz continuous payoff functions. One immediate interpretation of the uniform Lipschitz property of the functional $g\mapsto V(t,g)$ is that the option price is 
stable with respect to small perturbations in the initial curve $g$. This means in practical terms that the 
option price is robust towards small errors in the specification of the initial curve. It is important to notice
that we only have available a discrete set of forward prices in practice, and thus the specification of the 
initial curve $g$ may be prone to error as it is not perfectly observable.

Another interesting application of Prop.~\ref{prop:V-lipschitz} is the majorization of the option pricing error in case we 
wish to compute the price for a finite dimensional projection of the infinite dimensional curve $g$.
Recall that from a practical market perspective, we only have knowledge of a finite subset of values from 
the whole curve $g$. This is the situation we discuss now:

Let $\{e_k\}_{k\in\mathbb{N}}$ be an orthonormal basis of $H_{\alpha}$, and define the projection operator $\Gamma_n:H_{\alpha}\rightarrow H^n_{\alpha}$
by
\begin{equation}
\Gamma_ng=\sum_{k=1}^n\langle g,e_k\rangle_{\alpha}e_k\,,
\end{equation}
where $H^n_{\alpha}$ is the $n$-dimensional subspace of $H_{\alpha}$ spanned by the basis $\{e_1,\ldots,e_n\}$. The 
option price with $\Gamma_ng$ as initial curve becomes $V_n(t,g):=V(t,\Gamma_ng)$, and we find from 
Prop.~\ref{prop:V-lipschitz} that 
$$
\sup_{t\leq\tau}|V(t,g)-V_n(t,g)|\leq C\|g-\Gamma_ng\|_{\alpha}\,.
$$
But, when $n\rightarrow\infty$ it follows from Parseval's identity
$$
\|g-\Gamma_ng\|_{\alpha}^2=\sum_{k=n+1}^{\infty}|\langle g,e_k\rangle_{\alpha}|^2\rightarrow 0\,,
$$
and we can approximate $V(t,g)$ within a desirable error by choosing $n$ sufficiently big. Note that with
$$
\widehat{V}(t,x_1,\ldots,x_n):=V(t,\sum_{k=1}^nx_ke_k)
$$ 
we have that $V_n(t,g)=\widehat{V}(t,\langle g,e_1\rangle_{\alpha},\ldots,\langle g,e_n\rangle_{\alpha})$. We can
view $\widehat{V}(t,x_1,\ldots,x_n)$ as the option price on the $H_{\alpha}$-valued stochastic process $g$ which is 
started in the finite-dimensional subspace $H_{\alpha}^n$ at time $t$ with the 
values $\langle\Gamma_ng,e_k\rangle_{\alpha}=x_k, k=1,\ldots,n$. By the dynamics of $g$ we have no guarantee
that the process $g$ will remain in $H_{\alpha}^n$, so that at time $\tau$ we have in general that 
$g^{t,\Gamma_ng}(\tau)\notin H_{\alpha}^n$. Indeed, it may truly be an infinite dimensional object and thus not
in any $H_{\alpha}^m$, $m\in\N$. Furthermore, it is important to note that such an approximation
$\Gamma_n g$ typically fails to be a martingale under the pricing measure $Q$, and hence the option price $V_n(t,g)$ will not
be arbitrage-free. In a forthcoming paper~\cite{BK-martingale}, we study arbitrage-free finite dimensional approximations.

\subsection{The arithmetic Gaussian case}
Suppose that $g$ solves the simple linear Musiela equation
\begin{equation}
\label{eq:spde-arithmetic}
dg(t)=\partial_x g(t)\,dt+\sigma(t)\,d\mathbb{B}(t)
\end{equation}
where $\mathbb{B}$ is an $H$-valued Wiener process with covariance operator $\mathcal{Q}$
and $H$ being a separable Hilbert space. 
The volatility $\sigma$ is assumed to be a stochastic process 
$\sigma:\mathbb{R}_+\mapsto L(H, H_{\alpha})$, where  
$\sigma\in\mathcal{L}^2_{\mathbb{B}}(H_{\alpha})$, 
the space of integrands for the stochastic integral with respect to $\mathbb{B}$ 
(see Sect.~8.2 in Peszat and Zabczyk~\cite{peszat.zabczyk.07}). This is indeed
a special case of the general Markovian dynamics presented above, and the mild solution becomes 
\begin{equation}
g(\tau)=\mathcal{S}_{\tau-t}g(t)+\int_t^{\tau}\mathcal{S}_{\tau-s}\sigma(s)\,d\mathbb{B}(s)\,,
\end{equation} 
for $\tau\geq t$. We now analyse $V(t)$ defined in \eqref{option-val-formula} and 
\eqref{option-val-formula-x=0} for this particular 
dynamics. First recall from Lemma \ref{lemma:forward-dynamics} that
 \begin{align*}
 F(\tau,T_1,T_2) = \delta_{T_1-t}\mathcal D_\ell^w g(t)+\int_t^\tau\delta_{T_1-s}\mathcal D_\ell^w\sigma(s)\,d\mathbb{B}(s) 
 \end{align*}
 for any $t\in[0,\tau]$.

It follows from Theorem 2.1 in Benth and Kr\"uhner~\cite{BK-HJM} that 
\begin{align}\label{eq:general-price}
 F(\tau,T_1,T_2) &= \delta_{T_1-t}\mathcal D_\ell^w g(t)+\int_t^\tau\widetilde \sigma(s)\,dB(s)
\end{align}
for any $t\in[0,\tau]$ where
$$\widetilde\sigma^2(s) = (\delta_{T_1-s}\mathcal D_\ell^w\sigma(s)\mathcal{Q}\sigma^*(s)(\delta_{T_1-s}\mathcal D_\ell^w)^*)(1)\,$$
and $B$ is a standard Brownian motion.

This implies
 \begin{align*}
  V(t,g(t)) &= e^{-r(\tau-t)} \mathbb E[p(F(\tau,T_1,T_2))]
 \end{align*}
 
We find the following particular result for $V$ in the case of a non-random volatility:
\begin{proposition}\label{prop:general-price}
Let $\sigma$ be non-random. Then we have, 
$$
V(t,g)=\e^{-r(\tau-t)}\E\left[p(m(g)+\xi X)\right]\,.
$$
for any for $t\leq\tau\leq T_1$. Here, $X$ is a standard normal distributed random variable,
$$
\xi^2:=\int_t^{\tau}\widetilde{\sigma}^2(s)\,ds=\int_t^{\tau}(\delta_{T_1-s}\mathcal D_\ell^w\sigma(s)\mathcal{Q}\sigma^*(s)(\delta_{T_1-s}\mathcal D_\ell^w)^*)(1)\,ds\,,
$$
for any $t\in[0,\tau]$ and 
$$
m(g):=(\delta_{T_1-t}\circ\mathcal{D}^w_{\ell})(g)\,,\quad g\in H_\alpha.
$$
\end{proposition}
\begin{proof}
In the case of $\sigma$ being non-random, we find that the stochastic integral
$\int_t^{\tau}\widetilde{\sigma}(s)\,dB(s)$ in \eqref{eq:general-price} is a
centered normal distributed random variable. The variance is $\xi^2$, which follows 
straightforwardly by the It\^o isometry. 
By the independent increment property of Brownian motion, the result follows.
\end{proof}

In order to compute the realized variance $\xi^2$ in the Proposition above, we must find
the dual operator of $\delta_{T_1-s}\circ\mathcal D^w_{\ell}$. Obviously it holds that
$$
(\delta_{T_1-s}\circ\mathcal D^w_{\ell})^*=\mathcal{D}^{w*}_{\ell}\circ\delta^*_{T_1-s}\,.
$$
The dual operator of $\delta_y$ is found in Filipovic~\cite{filipovic} (see also Lemma~3.1
in \cite{BK-HJM}), and is the mapping $\delta_y^*:\R\mapsto H_{\alpha}$
defined as
\begin{equation}
\label{dual-delta}
\delta_y^*(c):x\mapsto c+c\int_0^{y\wedge x}\alpha^{-1}(u)\,du:=ch_y(x)
\end{equation}
for $c\in\R$ and $x\geq 0$ and 
\begin{equation}
\label{hx-function}
h_y(x)=1+\int_0^{y\wedge x}\alpha^{-1}(u)\,du\,.
\end{equation} 
Thus, $\delta_{T_1-s}^*(1)$ is the function
\begin{equation}
\delta_{T_1-s}^*(1)(x)=h_{T_1-s}(x)=1+\int_0^{(T_1-s)\wedge x}\alpha^{-1}(u)\,du\,,
\end{equation}
for $x\geq 0$. Now we are left to derive the function $\mathcal{D}^{w*}_{\ell}(h_{T_1-s})$.

\begin{proposition}
 With the preceding notations we have
  $$ \mathcal{D}^{w*}_{\ell}(h_{T_1-s})(x) = W_\ell(\ell)h_{T_1-s}(x) + \int_0^x \frac{q_\ell^w(T_1-s,z)}{\alpha(z)}dz $$
  for any $s\in[0,T_1]$, $x\geq0$.
\end{proposition}
\begin{proof}
  Let $x\geq0$ and $s\in[0,T_1]$. Then we have
  \begin{align*}
     \mathcal{D}^{w*}_{\ell}(h_{T_1-s})(x) &= \< \mathcal{D}^{w*}_{\ell}(h_{T_1-s}), h_x\> \\
                                           &= \< h_{T_1-s}, \mathcal{D}^{w}_{\ell}h_x\> \\
                                           &= \mathcal{D}^{w}_{\ell}h_x(T_1-s) \\
                                           &= W_\ell(\ell)h_{T_1-s}(x) + \int_0^x \frac{q_\ell^w(T_1-s,z)}{\alpha(z)}dz.
  \end{align*}
\end{proof}

If we define $\Sigma(s):= \widetilde\sigma(s)Q\widetilde\sigma^*(s)$ and if we want to apply Proposition \ref{prop:general-price}, then we need to calculate
 $$ \zeta^2 = \int_t^{\tau}(\delta_{T_1-s}\mathcal D_\ell^w)\Sigma(s)(\delta_{T_1-s}\mathcal D_\ell^w)^*(1)ds. $$
With a representation for $(\delta_{T_1-s}\mathcal D_\ell^w)^*(1)=\mathcal{D}^{w*}_{\ell}(h_{T_1-s})$ at hand we still need to calculate the operator $\delta_{T_1-s}\mathcal D_\ell^w$. However, we simply have
 $$ \delta_{T_1-s}\mathcal D_\ell^wg = W_\ell(\ell)g(T_1-s)+\int_0^{\infty}q_\ell^w(T_1-s,y)g'(y)dy $$
for any $g\in H_\alpha$. $Q$ and $\sigma$ are -- of course -- up to the modellers choice. However, after $\sigma$ and $Q$ have been picked one does need to calculate $\sigma^*(s)$. The following proposition gives a simple formula for calculating the dual operator of a given operator. As a side remark, the next proposition also shows that any linear operator $\mathcal T=(\mathcal T^*)^*$ on $H_\alpha$ is the sum of an integral operator and an operator which 'only' acts on the initial value of the inserted function.
\begin{proposition}
\label{l:integraldarstellung}
Let $\mathcal T\in L(H_\alpha)$. Then
$$ 
\mathcal T^*g(x) = g(0)\eta(x) + \int_0^\infty q(x,y)g'(y)dy,\quad g\in H_\alpha\,, 
$$
where
\begin{align*}
\eta(x) := (\mathcal Th_x)(0) = (\mathcal T^*h_0)(x),\\
q(x,y) := (\mathcal Th_x)'(y)\alpha(y)\,,
\end{align*}
for any $x,y\geq0$ and $h_x$ is defined in \eqref{hx-function}.
\end{proposition}
\begin{proof}
Filipovic~\cite[Lemma 5.3.1]{filipovic} shows that $g(x) = \<g,h_x\>$ for any $g\in H_\alpha$, $x\geq0$. Hence
\begin{align*}
       \mathcal T^*g(x) &= \<\mathcal T^*g,h_x\> \\
                      &= \<g,\mathcal Th_x\> \\
                      &= g(0)\mathcal Th_x(0) + \int_0^\infty g'(y) (\mathcal Th_x)'(y)\alpha(y)\,dy \\
                      &= g(0)\eta(x) + \int_0^\infty q(x,y) g'(y)\,dy\,,
\end{align*}
for any $g\in H_\alpha$, $x\geq 0$. This proves the result.
\end{proof}

Let us next move our attention to the so-called "delta" of the option price in Prop.~\ref{prop:general-price}. We 
define the "delta" to be the G\^ateaux derivative of the price $V(t,g(t))$ along some direction
$h\in H_{\alpha}$. This will measure how sensitive the price functional is to perturbations
along $h$ of the forward curve $g(t)$.  We have the following result:
\begin{proposition}
Assume $\sigma$ is non-random. Then the G\^ateaux derivative of $V(t,g(t))$ in
direction $h\in H_{\alpha}$ is
$$
D_hV(t,g)=\frac1{\xi}m(h)\E\left[p(m(g)+\xi X) X\right]
$$
with $m(g)$ and $\xi$ defined in Prop.~\ref{prop:general-price}.  
\end{proposition}  
\begin{proof}
We apply the so-called density method (see Glasserman~\cite{G}) along with 
properties of the G\^ateaux derivative.
For $g\in H_{\alpha}$, it holds after a change of variables,
\begin{align*}
V(t,g)&=\int_{\R}p(m(g)+\xi x)\phi(x)\,dx
=\frac1{\xi}\int_{\R}p(y)\phi\left(\frac{y-m(g)}{\xi}\right)\,dy\,,
\end{align*}
where $\phi$ is the standard normal probability density function. 
By the linear growth of $p$ and integrability properties of the normal
density function $\phi$, it follows that 
\begin{align*}
D_hV(t,g)&=\frac1{\xi}\int_{\R}p(y)D_h\phi\left(\frac{y-m(g)}{\xi}\right)\,dy \\
&=\frac1{\xi}\int_{\R}p(y)\phi'\left(\frac{y-m(g)}{\xi}\right)\left(-\frac1{\xi}\right)
D_hm(g)\,dy \\
&=\frac{1}{\xi^2}D_hm(g)\int_{\R}p(y)\left(\frac{y-m(g)}{\xi}\right)
\phi\left(\frac{y-m(g)}{\xi}\right)\,dy \\
&=\frac1{\xi}D_hm(g)\int_{\R}p(m(g)+\xi x)x\phi(x)\,dx \\
&=\frac{1}{\xi}D_hm(g)\E\left[p(m(g)+\xi X) X\right]\,.
\end{align*}
But obviously 
$$
D_hm(g)=\frac{d}{d\epsilon}m(g+\epsilon h)_{\epsilon=0}=\frac{d}{d\epsilon}(m(g)+\epsilon m(h))_{\epsilon=0}=m(h)\,,
$$
and the Proposition follows.
\end{proof}

It is interesting to note here that the delta computed in the Proposition above gives the sensitivity of the option
price to perturbations in the direction of a "forward curve" $h$. As mentioned earlier, the market only quotes forward
prices for a finite set of delivery periods, and not for all delivery times. Hence, we do not have the forward curve 
accessible. Indeed, we do not know $g(t)$ at time $t$, but only a finite set of values of swap prices, which 
is equivalent to a finite set of linear functionals on integral operators applied to $g$. 
It is market practice to "extract" such a curve by appealing to some smoothing 
techniques (see for example Benth, Koekebakker and Ollmar~\cite{BKO} for a spline 
approach). From given observations of delivery-period swap prices, one constructs
a forward curve of continuous delivery times. This will then give "the observed curve" $g(t)$ at time $t$. 
Remark that we need to have this curve accessible to price the option at time $t$, as we can see from Prop.~\ref{prop:general-price}. The extraction of such a curve from observations is by far a uniquely defined 
object (one can choose several different ways to produce such a curve), 
and as such it is crucial to use the expression for the delta to see how sensitive the price is towards perturbations of 
it.

We find the following explicit result for the price and sensitivity (delta) of call options:
\begin{proposition}
The price of a call option with strike $K$ and exercise time $\tau\leq T_1$ is
$$
V(t,g(t))=\xi\phi\left(\frac{m(g(t))-K}{\xi}\right)+(m(g(t))-K)\Phi
\left(\frac{m(g(t))-K}{\xi}\right)\,,
$$
where $\xi$ and $m(g)$ are defined in Prop.~\ref{prop:general-price}, 
$\Phi(x)$ is the cumulative standard normal distribution function and 
$\phi$ its density, i.e. $\Phi'(x)=\phi(x)$. Moreover,
$$
D_hV(t,g(t))=m(h)\Phi\left(\frac{m(g(t))-K}{\xi}\right)\,,
$$
for any $h\in H_{\alpha}$. 
\end{proposition}
\begin{proof}
For a call option with strike $K$ we have $p(F)=\max(F-K,0)$. Hence,
from Prop.~\ref{prop:general-price}
$$
V(t,g(t))=\int_{\R}\max\left(m(g(t))+\xi x-K,0\right)\phi(x)\,dx\,.
$$
The formula for $V(t,g(t))$ follows from standard calculations using the properties 
of the normal distribution. As for the G\^ateaux derivative of $V$, we calculate
this directly from $V(t,g(t))$. 
\end{proof}
Note that the expression for the sensitivity of $V$ with respect to $g$ is the classical "delta" of
a call option, scaled by $m(h)$. 

As a slight extension of the option pricing theory above, we discuss a class of spread options written on forwards with 
different delivery periods. To this end, consider an option written on two forwards with delivery periods
being $[T_1^1,T_2^1]$ and $[T_1^2,T_2^2]$ respectively, where the option pays
$p(F(\tau,T_1^1,T_2^1),F(\tau,T_1^2,T_2^2))$ at exercise time $\tau\leq\min(T_1^1,T_1^2)$.
We assume that $p:\R^2\rightarrow\R$ is a measurable function of at most linear growth. 
For example, $p(x,y)=\max(x-y,0)$ will be the payoff from the spread between two forwards of different
delivery periods, a kind of calendar spread option. By following the arguments
of Prop.~\ref{prop:payoff-repr} we find that
\begin{equation}
p((F(\tau,T_1^1,T_2^1),F(\tau,T_1^2,T_2^2))=\mathcal{P}_{\ell_1,\ell_2}(T_1^1-\tau,T_1^2-\tau,g(\tau))\,,
\end{equation}
for $\mathcal{P}_{\ell_1,\ell_2}:\R_+^2\times H_{\alpha}\rightarrow\R$ defined as
\begin{equation}
\mathcal{P}_{\ell_1,\ell_2}(x,y,g)=p\circ(\delta_x\circ\mathcal{D}_{\ell_1}^w(g),\delta_y\circ\mathcal{D}_{\ell_2}^w(g))\,.
\end{equation}
Here, $\ell_i=T_2^i-T_1^i$, $i=1,2$. By the linear growth of $p$, we can show that $\mathcal{P}_{\ell_1,\ell_2}$
is at most linearly growing in $\|g\|_{\alpha}$, uniformly in $x,y$. By following the arguments for the univariate case 
above, the price of the option at time $t\leq\tau$ can
be computed as follows:
\begin{align*}
V(t,g(t))&=\e^{-r(\tau-t)}\E\left[p\left(\delta_{T_1^1-\tau}\circ\mathcal{D}_{\ell_1}^w(g(\tau)),
\delta_{T_1^2-\tau}\circ\mathcal{D}_{\ell_2}^w(g(\tau))\right)\,|\,\mathcal{F}_t\right] \\
&=\e^{-r(\tau-t)}\E\left[p(F(\tau,T_1^1,T_2^1),F(\tau,T_1^2,T_2^2))\,|\,g(t)\right] \,.
\end{align*}
Yet again, we find
 \begin{align*}
  (F(\tau,T_1^1,T_2^1),F(\tau,T_1^2,T_2^2)) &= (\delta_{T_1^1-\tau}\circ\mathcal{D}_{\ell_1}^w(g(\tau)),\delta_{T_1^2-\tau}\circ\mathcal{D}_{\ell_2}^w(g(\tau))) \\
                        &= (\delta_{T_1^1-t}\circ\mathcal{D}_{\ell_1}^w(g(t)),\delta_{T_1^2-t}\circ\mathcal{D}_{\ell_2}^w(g(t))) \\
                        &\ + \int_t^\tau (\delta_{T_1^1-s}\mathcal{D}_{\ell_1}^w\sigma(s)\,d\mathbb{B}(s),\delta_{T_1^2-t}\mathcal{D}_{\ell_2}^w\sigma(s)\,d\mathbb{B}(s)) \\
                        &= (\delta_{T_1^1-t}\circ\mathcal{D}_{\ell_1}^w(g(t)),\delta_{T_1^2-t}\circ\mathcal{D}_{\ell_2}^w(g(t)))  + \int_t^\tau \Sigma(s)dB(s) \\
 \end{align*}
where $B$ is some two-dimensional standard Brownian motion and $\Sigma(s)$ is the positive semidefinite root of
 $$ \big((\delta_{T_1^i-s}\mathcal{D}_{\ell_i}^w\sigma(s))Q(\delta_{T_1^j-s}\mathcal{D}_{\ell_j}^w\sigma(s))^*\big)_{i,j=1,2}$$
for any $s\geq 0$. The matrix $\Sigma^2(s)$ can be computed as before and appears in the formula for the realized variance. Hence,
$$ V(t,g) = \mathbb E\left[p\left((\delta_{T_1^1-t}\circ\mathcal{D}_{\ell_1}^w(g),\delta_{T_1^2-t}\circ\mathcal{D}_{\ell_2}^w(g))  + \int_t^\tau \Sigma(s)dB(s)\right)\right]$$
for any $t\in[0,\tau]$, $g\in H_\alpha$.

In conclusion, we see that we get a two-dimensional stochastic It\^o integral of a deterministic 
integrand in the expectation defining the price $V(t,g)$, yielding a bivariate Gaussian random variable. 
Therefore we can -- after computing the correlation -- represent the
option price as an expectation of a function of a bivariate Gaussian random variable. The correlation will depend on
$\mathcal{Q}$, the spatial covariance structure of the noise $\mathbb{B}$, the "volatility" $\sigma(s)$ of the 
forward curve $\sigma$, as well as the delivery periods of the two forwards. Roughly explained, we are extracting two 
pieces of the forward curve (defined by the delivery periods), and constructing a bivariate Gaussian random variable of it.
Although the expression involved becomes rather technical, we can obtain rather explicit option prices which honour
the spatial dependency structure of the forward curve.

\subsection{The geometric Gaussian case}
First, we show that the Hilbert space $H_{\alpha}$ is closed under
exponentiating:
\begin{lemma}
If $g\in H_{\alpha}$, then $\exp(g)\in H_{\alpha}$ where $\exp(g)=\sum_{n=0}^{\infty}g^n/n!$. 
\end{lemma} 
\begin{proof}
First, if $g\in H_{\alpha}$ then $x\mapsto \exp(g(x))$ is an absolutely continuous 
function from $\R_+$ into $\R_+$.
Due to Prop.~4.18 in Benth and Kr\"uhner~\cite{BK-HJM}, $H_{\alpha}$ is a Banach algebra
with respect to the norm $\|\cdot\|:=k_1\|\cdot\|_{\alpha}$, where $k_1=\sqrt{5+4k^2}$
and $k^2=\int_0^{\infty}\alpha^{-1}(x)\,dx$. I.e., if $f,g\in H_{\alpha}$, then
$\|fg\|\leq \|f\| \|g\|$. By the triangle inequality we therefore have 
$\|\exp(g)\|\leq \exp(\|g\|)<\infty$ for any $g\in H_{\alpha}$, or, in other words,
$$
\|\exp(g)\|_{\alpha}\leq\frac1{k_1}\exp(k_1\|g\|_{\alpha})<\infty\,.
$$ 
Hence, $\exp(g)\in H_{\alpha}$, and the Lemma follows. 
\end{proof}
Suppose that the forward prices are given as the exponential of a stochastic process
in $H_{\alpha}$, i.e., of the form
\begin{equation}
g(t)=\exp(\widetilde{g}(t))\,,
\end{equation}
where
\begin{equation}
\label{eq:spde-geometric}
d\widetilde{g}(t)=(\partial_x\widetilde{g}(t)+\mu(t))\,dt+\sigma(t)\,d\mathbb{B}(t)\,,
\end{equation}
where $\sigma$ and $\mathbb{B}$ are as for the stochastic partial differential equation in \eqref{eq:spde-arithmetic}, 
and $\mu$ a predictable $H_{\alpha}$-valued
stochastic process which is Bochner integrable on any finite time interval. 
To have a no-arbitrage dynamics, we must impose the drift condition (see Barth and Benth~\cite{BB})
\begin{equation}
\label{eq:drift-condition}
x\mapsto \mu(t,x):=-\frac12\|\delta_x\sigma(t)\mathcal{Q}^{1/2}\|_{L_{HS}(H,\R)}\,.
\end{equation}
We assume that this drift condition holds from now on. The following simplification of the drift
condition holds true:
\begin{lemma}
\label{lemma:drift-condition}
The drift condition for $\mu$ in \eqref{eq:drift-condition} can be expressed as
$$
\mu(t,x)=-\frac12\delta_x\sigma(t)\mathcal{Q}\sigma^*(t)\delta_x^*(1)\,.
$$
\end{lemma}
\begin{proof}
It follows from the definition of the Hilbert-Schmidt norm that
$$
\mu(t,x)=-\frac12\sum_{k=1}^{\infty}(\delta_x\sigma(t)\mathcal{Q}^{1/2}e_k)^2\,,
$$
where $\{e_k\}_k$ is a basis of $H$. But,
\begin{align*}
(\delta_x\sigma(t)\mathcal{Q}^{1/2})(e_k)\cdot1=\langle e_k,(\delta_x\sigma(t)\mathcal{Q}^{1/2})^*(1)\rangle_H 
=\langle e_k,\mathcal{Q}^{1/2}\sigma^*(t)\delta_x^*(1)\rangle_H\,.
\end{align*}
Hence, by linearity of operators,
\begin{align*}
\mu(t,x)&=-\frac12\sum_{k=1}^{\infty}(\delta_x\sigma(t)\mathcal{Q}^{1/2}e_k)\langle e_k,\mathcal{Q}^{1/2}\sigma^*(t)\delta_x^*(1)\rangle_H \\
&=\delta_x\sigma(t)\mathcal{Q}^{1/2}(\sum_{k=1}^{\infty}\langle e_k,\mathcal{Q}^{1/2}\sigma^*(t)\delta_x^*(1)\rangle_H e_k) \\
&=\delta_x\sigma(t)\mathcal{Q}^{1/2}(\mathcal{Q}^{1/2}\sigma^*(t)\delta_x^*(1))\,.
\end{align*}
The result follows.
\end{proof}
We recall that $\delta_x^*(1)=h_x$, with the function $y\mapsto h_x(y)$ is defined in \eqref{hx-function}.  Thus, we can write $\mu(t,x)=-\delta_x\sigma(t)\mathcal{Q}\sigma^*(t)h_x(\cdot)/2$. 


As for \eqref{eq:spde-arithmetic} in the subsection above, we have a mild solution of the stochastic partial differential equation 
\eqref{eq:spde-geometric} satisfying 
for $\tau\geq t$
$$
\widetilde{g}(\tau)=\mathcal{S}_{\tau-t}\widetilde{g}(t)+
\int_t^{\tau}\mathcal{S}_{\tau-s}\mu(s)\,ds
+\int_t^{\tau}\mathcal{S}_{\tau-s}\sigma(s)\,d\mathbb{B}(s)\,.
$$

The following lemma states the dynamics of the curve valued process $g(t):=\exp(\widetilde g(t))$, $t\geq0$, revealing 
that $g$ is Markovian as in Section \ref{sect:Markovian}.
\begin{lemma}
 Under the drift condition \eqref{eq:drift-condition} we have
  $$ g(\tau) = \mathcal{S}_{\tau-t}g(t) + \int_t^\tau \mathcal{S}_{\tau-s}\widehat\sigma(s,g(s))\,d\mathbb{B}(s)\,$$
 for any $0\leq t\leq \tau$ where $\widehat\sigma(s,g)h(x):=g(x)\sigma(s)h(x)$ for any $x\geq0$, $g,h\in H_\alpha$.
  Consequently, the forward dynamics are given by
   $$ F(\tau,T_1,T_2) = \delta_{T_1-t}\mathcal D_\ell^wg(t) + \int_t^\tau \delta_{T_1-s}\mathcal D_\ell^w\widehat\sigma(s,g(s))\,d\mathbb{B}(s).$$
\end{lemma}
\begin{proof}
  Recall that $G(\tau,T) = g(\tau)(T-\tau)$ and define $\widetilde G(\tau,T):=\widetilde g(\tau)(T-\tau)$. Then we have
  \begin{align*}
     G(\tau,T) &= g(\tau)(T-\tau) \\
               &= \exp(\widetilde g(\tau)(T-\tau)) \\
               &= \exp(\widetilde G(\tau,T))
  \end{align*}
  for any $0\leq \tau\leq T$. Moreover, we have
   $$\widetilde G(\tau,T) = \delta_{T-t}\widetilde g(t) + \int_t^\tau \delta_{T-s}(\mu(s)ds + \sigma(s)\,d\mathbb{B}(s)) $$
  and hence It\^o's formula together with the drift condition \eqref{eq:drift-condition} yields
   \begin{align*}
    G(\tau,T) &= \exp(\widetilde G(\tau,T)) \\
              &= \delta_{T-t} g(t) + \int_t^\tau G(s,T)\delta_{T-s}\sigma(s)\,d\mathbb{B}(s) \\
              &= \delta_{T-t} g(t) + \int_t^\tau \delta_{T-s}\widehat\sigma(s,g(s))\,d\mathbb{B}(s)    
   \end{align*}
   for any $0\leq \tau\leq T$. Since $g(\tau)(x) = G(\tau,\tau+x)$ we conclude that
    $$ g(\tau) = \mathcal S_{\tau-t}g(t) + \int_t^\tau \mathcal S_{\tau-s}\widehat\sigma(s,g(s))\,d\mathbb{B}(s) $$
   for any $0\leq t\leq \tau$.
\end{proof}

The price of a European option with exercise time $\tau\geq t$ on a forward delivering at time $T$ 
when $\sigma$ is non-random
can be easily derived as in the arithmetic case. Indeed, it holds that 
\begin{equation}
V(t,\widetilde g)=\e^{-r(\tau-t)}\mathbb{E}\left[p(\exp(\widehat{m}(\widetilde g)+\xi X))\right]
\end{equation}
where  $X$ is a standard normal distributed random variable, $\xi$ is as in Prop.~\ref{prop:general-price} 
(using the $T$ instead of $T_1$) and
\begin{equation}
\widehat{m}(g)=\widetilde g(T-t)-\frac12\int_t^{\tau}\mu(s)(T-s)\,ds
\end{equation}
If we let $p$ be the payoff function of a call option, then a simple calculation shows that 
we recover the Black-76 formula (see Black~\cite{black}, or Benth, \v{S}altyt\.{e} Benth and 
Koekebakker~\cite{BSBK-energy} for a more general version). 

Finally, we remark that if we are interested in pricing
options written on a forward delivering over a period, the payoff function will become
$$
p((\delta_{T-\tau}\circ\mathcal{D}^w_{\ell})(g(\tau)))=
p(F(\tau,T_1,T_2))\,.
$$
The integral operator $\mathcal{D}^w_{\ell}$ maps $\exp(\widetilde{g}(\tau))\in H_{\alpha}$
into $H_{\alpha}$, however, we do not have any nice representation of it. The
problem is of course that the integral of the exponent of a general function is not 
analytically known. Thus, it seems difficult to obtain any tractable expression yielding simple pricing formulas.

\subsection{L\'evy models}
We include a brief discussion on the pricing of options when the forward curve is driven by a L\'evy process
$\mathbb{L}$. We confine our analysis to the arithmetic model
\begin{equation}
\label{eq:arithmetic-levy-spde}
dg(t)=\partial_xg(t)\,dt+\sigma(t)\,d\mathbb{L}(t)\,,
\end{equation} 
where $\mathbb{L}$ is a L\'evy process with values in a separable Hilbert space
$H$, having zero mean and being square integrable. The stochastic process $\sigma:\R_+\rightarrow L(H,H_{\alpha})$ is
integrable with respect to $\mathcal{L}$, i.e., $\sigma\in\mathcal{L}^2_{\mathbb{L}}(H_{\alpha})$
(see Sect.~8.2 in Peszat and Zabczyk~\cite{peszat.zabczyk.07} for this notation.) 

The price of an option given in \eqref{option-val-formula} requires the computation of 
$(\delta_{T_1-\tau}\circ\mathcal{D}_{\ell}^w)(g(\tau))$. As for the Gaussian models,
there exists a mild solution of \eqref{eq:arithmetic-levy-spde} which for $\tau\geq t\geq 0$ is given by
\begin{equation}
g(\tau)=\mathcal{S}_{\tau-t}g(t)+\int_t^{\tau}\mathcal{S}_{\tau-s}\sigma(s)\,d\mathbb{L}(s)\,.
\end{equation}
From the linearity of the operators, it holds
$$
(\delta_{T_1-\tau}\circ\mathcal{D}_{\ell}^w)g(\tau)=(\delta_{T_1-t}\circ\mathcal{D}_{\ell}^w)g(t)
+\int_t^{\tau}(\delta_{T_1-s}\circ\mathcal{D}_{\ell}^w)\sigma(s)\,d\mathbb{L}(s))\,.
$$
The first term on the right hand side is, not surprisingly, $m(g(t))$ with $m$ defined in Prop.~\ref{prop:general-price}.
For the Gaussian model, we used a result in Benth and Kr\"uhner~\cite{BK-HJM} that provided us with an explicit
representation of a linear functional applied on a $H_{\alpha}$-valued stochastic integral with respect to
a $H$-valued Wiener process. One can write this functional as a stochastic integral of a real-valued stochastic integrand 
with respect to a real-valued Brownian motion. The integrand is, moreover, explicitly known. Something similar is known
for the special class of L\'evy processes being subordinated Wiener processes.  

Following Benth and Kr\"uhner \cite{BK}, we introduce $H$-valued subordinated Brownian motion:
Denote by $U(t)\}_{t\geq 0}$ a L\'evy process with values on the positive real line, that is, a
non-decreasing L\'evy process. These processes are frequently called {\it subordinators} (see Sato~\cite{Sato}).
Let $\mathbb{L}(t):=\mathbb{B}(U(t))$, which then becomes a L\'evy process with values in $H$.
In Benth and Kr\"uhner~\cite{BK} one finds conditions on $U$ implying that $\mathbb{L}$ is a
zero-mean square integrable L\'evy process. 

From Thm.~2.5 in Benth and Kr\"uhner~\cite{BK-HJM}, we find that 
$$
\int_t^{\tau}(\delta_{T_1-s}\circ\mathcal{D}_{\ell}^w)\sigma(s)\,d\mathbb{L}(s)) =\int_t^{\tau}\widetilde{\sigma}(s)\,dL(s)\,,
$$
where $L$ is a real-valued subordinated Brownian motion $L(t):=B(U(t))$, $B$ being a standard 
Brownian motion. Moreover, the process $\widetilde{\sigma}(s)$ is given by
$$
\widetilde{\sigma}^2(s)=(\delta_{T_1-s}\circ\mathcal{D}_{\ell}^w)\sigma(s)\mathcal{Q}
\sigma^*(s)(\delta_{T_1-s}\circ\mathcal{D}_{\ell}^w)^*(1)\,,
$$ 
which is identical to the Gaussian case studied above.

For the problem of pricing options, we see that we are back to computing the expectation of a functional of a univariate 
stochastic integral. If $\sigma$ is non-random, we can use for example Fourier techniques to 
compute this expectation, as we know the cumulant function of $L$ from the cumulant of $U$ and 
Brownian motion  (see Carr and Madan \cite{carrmadan} an account on Fourier methods in derivatives pricing, and Benth, \v{S}altyt\.{e} Benth and Koekebakker \cite{BSBK-energy} for the application to energy markets). 
This will provide us with an expression for the option price that can be efficiently computed using
fast Fourier transform techniques. 

We end with an example on a subordinated L\'evy process of particular interest in energy markets.  Assume
$U$ is an inverse Gaussian subordinator, that is, a L\'evy process with non-decreasing paths and $U(1)$ is inverse Gaussian
distributed. Then $\mathbb{L}(t)=\mathbb{B}(U(t))$ becomes an $H$-valued normal inverse Gaussian (NIG)
L\'evy process in the sense defined by  Benth and Kr\"uhner~\cite[Def.~4.1]{BK}. In fact,
for any functional $\mathcal{L}\in L(H,\R^n)$, $t\mapsto\mathcal{L}(\mathbb{L}(t))$ will be an $n$-variate 
NIG L\'evy process, with the particular case $L(t)$ introduced above defining an NIG 
L\'evy process on the real line. We refer to Barndorff-Nielsen~\cite{BN} for details on the inverse Gaussian 
subordinator and NIG L\'evy processes. Several empirical studies have demonstrated that returns of energy forward
and futures prices can be conveniently modelled by the NIG distribution (see Benth, \v{S}altyt\.{e} Benth and
Koekebakker~\cite{BSBK-energy} and the references therein for the case of NordPool power prices).  
Frestad, Benth and Koekabkker~\cite{FBK} and Andresen, Koekebakker and Westgaard~\cite{AKW} 
find that the NIG distribution fits power 
forward returns with
fixed time to maturity and given delivery period. Their analysis cover time series of prices with different times to maturity
and different delivery periods (weekly, monthly, quarterly, say), where these time series are constructed from a 
non-parametric smoothing of the original price data observed in the market. In fact, in our modelling context,
they are looking at time series observations of the stochastic process $t\mapsto(\delta_x\circ\mathcal{D}_{\ell}^w)(g(t))$.
From the analysis above we see that choosing $\mathbb{L}$ to be an $H$-valued NIG L\'evy process and $g$ being
an arithmetic dynamics will give
price increments being NIG distributed. Of course, this is not the same as the returns being NIG. As we have mentioned
earlier, it is not straightforward to model the price of forward with delivery period using an exponential dynamics. 
Frestad, Benth and Koekebakker~\cite{FBK}, and Andresen, Koekebakker and Westgaard~\cite{AKW} also estimate empirically the volatility term structure and the 
spatial (in time to maturity) correlation structure,
which provides information on the volatility $\sigma(t)$ and the covariance operator $\mathcal{Q}$. Indeed,
Andresen, Koekebakker and Westgaard~\cite{AKW} propose a multivariate NIG distribution to model the returns.

%

\section{Cross-commodity modelling}

In this Section we want to analyse a joint model for the forward curve evolution in
two commodity markets. For example, European power markets are inter-connected, and thus
forward prices will be dependent. Also, the markets for gas and coal will influence the power market, since
gas and coal are important fuels for power generation in many countries like for example UK and Germany.
This links forward contracts on gas and coal to those traded in the power markets. Finally, weather clearly affects
the demand (through temperature) and supply (through precipitation and wind) of energy, and one
can therefore also claim a dependency between weather futures (traded at Chicago Mercantile
Exchange (CME), say) and power futures.  These examples motivate the introduction of 
multivariate dynamic models for the time evolution of forward curves across different markets. We will restrict
our attention merely to the bivariate case here, and make some detailed analysis of a two-dimensional
forward curve dynamics. 


Consider two commodity forward markets. We model the "bivariate" forward curve dynamics 
$t\mapsto(g_1(t),g_2(t))$ as the $H_{\alpha}\times H_{\alpha}$-valued stochastic process being the
solution of the SPDE 
\begin{align}
dg_1(t)&=\partial_x g_1(t)\,dt+\sigma_1(t,g_1(t),g_2(t))\,d\mathbb{L}_{1}(t) \nonumber\\
dg_2(t)&=\partial_x g_2(t)\,dt+\sigma_2(t,g_1(t),g_2(t))\,d\mathbb{L}_{2}(t) \,,\label{eq:bivariate-hjm-dyn}
\end{align}
with $(g_1(0),g_2(0)=(g_1^0,g_2^0)\in H_{\alpha}\times H_{\alpha}$ given. We suppose that
$(\mathbb{L}_1,\mathbb{L}_2)$ is an $H_1\times H_2$-valued square-integrable zero-mean 
L\'evy process, where
$H_i, i=1,2$ are two separable Hilbert spaces and $\mathcal{Q}_i, i=1,2$ are the 
respective (marginal) covariance operators, i.e.\ $\mathbb E[\<L_i(t),g\>\<L_i(s),h\>]=(t\wedge s)\<Q_ig,h\>$ for any $t,s\geq0$, $g,h\in H_\alpha$ and $i=1,2$. Furthermore, we assume that
$\sigma_i:\R_+\times H_{\alpha}\times H_{\alpha}\rightarrow L(H_i,H_{\alpha})$ for
$i=1,2$ are measurable functions, and that there exists an increasing function $K:\R_+\rightarrow\R_+$
such that $\sigma_i$, $i=1,2$ are Lipschitz and 
of linear growth, that is, for any $(f_1,f_2),(h_1,h_2)\in H_{\alpha}\times H_{\alpha}$ and $t\in\R_+$,
\begin{align}
\|\sigma_i(t,f_1,f_2)-\sigma_i(t,h_1,h_2)\|_{\text{op}}&\leq K(t)\|(f_1,f_2)-(h_1,h_2)\|_{H_{\alpha}\times H_{\alpha}}\,, \\
\|\sigma_i(t,f_1,f_2)\|_{\text{op}}&\leq K(t)(1+\|(f_1,f_2)\|_{H_{\alpha}\times H_{\alpha}})\,.
\end{align}
Note that since the product of two (separable) Hilbert space again is a (separable) Hilbert space (using the 
canonical $2$-norm, i.e.\ $\Vert (f,g)\|_{H_{\alpha}\times H_{\alpha}}^2:=\|f\|_{H_\alpha}^2+\|g\|_{H_\alpha}^2$), we can relate to the theory of existence and uniqueness of mild solutions of 
SPDEs given by Tappe~\cite{Tappe}: there exists a unique mild solution satisfying the integral equations
\begin{align}
g_1(t)&=\mathcal{S}_{t}g_1^0+\int_0^t\mathcal{S}_{t-s}\sigma_1(s,g_1(s),g_2(s))\,d\mathbb{L}_1(s) 
\nonumber\\
g_2(t)&=\mathcal{S}_{t}g_2^0+\int_0^t\mathcal{S}_{t-s}\sigma_2(s,g_1(s),g_2(s))\,d\mathbb{L}_2(s)\,.
\label{eq:bivariate-hjm-mild}
\end{align}
Observe that $t\mapsto(F_1(t,T),F_2(t,T)):=(\delta_{T-t}g_1(t),\delta_{T-t}g_2(t))$, $t\leq T$ will be an $H_{\alpha}\times H_{\alpha}$-valued (local) martingale.
Moreover, the marginal $H_{\alpha}$-valued processes $t\mapsto F_i(t,T):=\delta_{T-t}g_i(t), i=1,2, t\leq T$ will also be (local) martingales,
ensuring that we have an arbitrage-free model for the forward price dynamics in the two commodity
markets.

Our main concern in the rest of this Section is to analyse in detail the "bivariate" L\'evy process
$(\mathbb{L}_1,\mathbb{L}_2)$. We are interested in its probabilistic properties in terms of
representation of the covariance operator and linear decomposition. Since
 $(\mathbb{L}_1(1),\mathbb{L}_2(1))$ is an $H_1\times H_2$-valued square-integrable variable, we 
analyse general square-integrable random variables $(X_1,X_2)$ in $H_1\times H_2$.

%
%
Before we set off, we recall the spectral theorem for normal compact operators on Hilbert spaces 
(see e.g. \cite[Statement 7.6]{Conway}):
\begin{proposition}\label{P:funktionalkalkuel}
Let $H$ be a separable Hilbert space and $\mathcal{T}$ be a symmetric 
compact operator. Then there is an orthonormal basis $\{e_i\}_{i\in \N}$ of $H$ and 
a family $\{\lambda_i\}_{i\in \N}$ of real numbers such that
$$ 
\mathcal{T}f = \sum_{i\in \N}\lambda_i\langle e_i,f\rangle e_i\,,
$$
for any $f\in H$. Let $\phi:\mathbb R\rightarrow\mathbb R$ be measurable. Then
$$ 
\phi(\mathcal{T}):\left\{ f\in H: \sum_{i\in \N} \vert \phi(\lambda_i)\vert^2 \langle e_i,f\rangle ^2 <\infty\right\} \rightarrow H, f \mapsto \sum_{i\in \N}\phi(\lambda_i)\langle e_i,f\rangle e_i\,, 
$$
defines a closed linear symmetric operator which is bounded and everywhere defined 
if $\phi$ is bounded on $\{\lambda_i:i\in \N\}$. For measurable $\phi,\psi:\mathbb R\rightarrow\mathbb R$ with 
$\psi$ bounded we have 
$(\phi+\psi)(\mathcal{T}) = \phi(\mathcal{T}) + \psi(\mathcal{T})$ and 
$(\phi\psi)(\mathcal{T}) = \phi(\mathcal{T})\psi(\mathcal{T})$.
\end{proposition}
We will apply this result in particular to define the square-root and the pseudo-inverse of a 
compact operator. We shall use the definition of a pseudo-inverse given in Albert~\cite{albert.72}:
\begin{definition}
Let $\mathcal{P}$ be a positive semidefinite compact operator on a separable Hilbert space $H$. 
Then $\mathcal{R}:=\sqrt{\mathcal{P}}$ is the {\em square-root} of $\mathcal{P}$. 
The {\em pseudo-inverse} $\mathcal{J}$ of $\mathcal{P}$ is defined by 
$\mathcal{J}:=\phi(\mathcal{P})$ where 
$\phi:\mathbb R\rightarrow\mathbb R,x\mapsto 1_{\{x\neq 0\}}/x$.
\end{definition}

Next, we want to represent covariance operators of square-integrable random variables in $H_1\times H_2$ in terms of operators on $H_1$, $H_2$ and between thoose spaces. To this end, we will need the natural projectors $\Pi_1:H_1\times H_2\rightarrow H_1,(x,y)\mapsto x$ and $\Pi_2:H_1\times H_2\rightarrow H_2,(x,y)\mapsto y$. We have the following general statement on the representation of the covariance operator of square-integrable
random variables in $H_1\times H_2$:
\begin{theorem}\label{s:Notwendigkeit A}
For $i=1,2$, let $X_i$ be a square integrable $H_i$-valued 
random variable and $\mathcal{Q}_i$ its covariance operator. Denote the positive semidefinite square-root of $\mathcal{Q}_i$ by $\mathcal{R}_i$ for $i=1,2$.
Then there is a linear operator $\mathcal{Q}_{12}\in L(H_1,H_2)$ such that
 \begin{enumerate}
  \item[(i)] $ \mathcal{Q}:= \left(\begin{matrix}\mathcal{Q}_1&\mathcal{Q}_{12}^*\\\mathcal{Q}_{12}
&\mathcal{Q}_2\end{matrix}\right)$ is the covariance operator of the 
$H_1\times H_2$-valued square integrable random variable $(X_1,X_2)$,
  \item[(ii)] $\vert\<\mathcal{Q}_{12}u,v\>\vert\leq \Vert \mathcal{R}_1u\Vert_1\Vert \mathcal{R}_2v\Vert_2$ for any $u\in H_1$, $v\in H_2$ and
  \item[(iii)] $\mathrm{ran}(\mathcal Q_{12})\subseteq\overline{\mathrm{ran}(\mathcal Q_2)}$ and $\mathrm{ran}(\mathcal Q_{12}^*)\subseteq\overline{\mathrm{ran}(\mathcal Q_1)}$.
 \end{enumerate}
\end{theorem}
\begin{proof}

 (i) Let $\mathcal{Q}$ be the covariance operator of $(X_1,X_2)$ and $\Phi_i:H_i\rightarrow H_1\times H_2$ be the natural embedding, i.e.\ $\Phi_i=\Pi^*_i$ for $i=1,2$. Define
$$ 
\mathcal{Q}_{12}:=\Pi_2\mathcal{Q}\Phi_1\,.
$$ 
Then the first assertion is evident.

 (ii) Let $u\in H_1$, $v\in H_2$ and $\beta\in \mathbb R$. We have
  \begin{align*}
   0 &\leq \<\mathcal{Q}(\beta u,v),(\beta u,v)\> \\
     &= \beta^2\<\mathcal{Q}_1u,u\>+\<\mathcal{Q}_2v,v\>
+2\beta\<\mathcal{Q}_{12}u,v\> \\
     &= \beta^2\Vert \mathcal{R}_1u\Vert_1^2+\Vert \mathcal{R}_2v\Vert^2_2 +2\beta\<\mathcal{Q}_{12}u,v\>\,,
  \end{align*}
 and hence
 \begin{align*}
   -\beta\<\mathcal{Q}_{12}u,v\> \leq \frac{1}{2}\left(\beta^2\Vert \mathcal{R}_1u\Vert_1^2+\Vert \mathcal{R}_2v\Vert_2^2\right)\,.
 \end{align*}
  Letting $\beta$ have the same sign as $-\<\mathcal{Q}_{12}u,v\>$ yields
 \begin{align*}
   \vert\beta\vert\vert\<\mathcal{Q}_{12}u,v\>\vert \leq \frac{1}{2}\left(\beta^2\Vert \mathcal{R}_1u\Vert_1^2+\Vert \mathcal{R}_2v\Vert_2^2\right).
 \end{align*}
  If $\Vert \mathcal{R}_1u\Vert_1=0$, then with $\vert\beta\vert\rightarrow \infty$ we see that $\vert\<\mathcal{Q}_{12}u,v\>\vert=0$ and hence the claimed inequality holds. Thus we may assume that $\Vert \mathcal{R}_1u\Vert_1\neq0$. Choosing
$\beta=\frac{\Vert \mathcal{R}_2v\Vert_2}{\Vert \mathcal{R}_1u\Vert_1}$ yields
 $$ 
\vert\<\mathcal{Q}_{12}u,v\>\vert \leq \Vert \mathcal{R}_1u\Vert_1
\Vert \mathcal{R}_2v\Vert_2$$
 as claimed.
 
 (iii) We show that $\mathcal Q_{12}u$ is orthogonal to any $v\in \mathrm{Kern}(\mathcal Q_2)$ for any $u\in H_1$. If that is done, then the claim follows because $\mathcal Q_2$ is positive semidefinite and hence its kernel and the closure of its range are closed orthogonal spaces. Let $u\in H_1$, $v\in H_2$ such that $\mathcal Q_2v=0$. Then, $\mathcal R_2v=0$ and hence (ii) yields
 \begin{align*}
    \vert \<\mathcal Q_{12}u,v\>| \leq \|\mathcal R_1u|\|_1\|\mathcal R_2v\|_2 = 0.
 \end{align*}
 The corresponding arguments show that $\mathcal Q_{12}^*$ maps into the closure of the range of $\mathcal Q_1$.

\end{proof}
Consider now $H_i=H_{\alpha_i}$, $H_{\alpha_i}$ being the Filipovic space with weight function $\alpha_i$, 
$i=1,2$. We suppose that both weight functions $\alpha_1,\alpha_2$ satisfy the hypotheses stated at the
beginning of Section~\ref{sect:hilbert-space-realization}. 
We first demonstrate that the operator $\mathcal{Q}_{12}$ yields the covariance between $\mathbb{L}_1(t)$ and
$\mathbb{L}_2(t)$ evaluated at two different maturities $x$ and $y$, with $x,y\in\R_+$. To this end, recall the function $h_x$ in \eqref{hx-function}. Then we have for any $x\in\R_+$ and $X\in H_{\alpha}$,
$$
\delta_x X=\langle X,\delta^*_x(1)\rangle=\langle X,h_x\rangle\,,
$$
by \eqref{dual-delta}. 
Hence, with $h_x^i$ being the function $h_x$ defined in \eqref{hx-function} using 
the weight function $\alpha_i$,
\begin{align*}
\delta_z^i(\mathbb{L}_i(t))&=\langle\mathbb{L}_i(t),h_z^i\rangle\,.
\end{align*}
Thus, with $(\mathbb{L}_1,\mathbb{L}_2)$ being a zero mean 
L\'evy process, we find for $x,y\in\R_+$,
\begin{align*}   
\text{Cov}(\mathbb{L}_1(t,x),\mathbb{L}_2(t,y))&=\E\left[\delta_x^1(\mathbb{L}_1(t))\delta_y^2(\mathbb{L}_2(t))\right] \\
&=\E\left[\langle\mathbb{L}_1(t),h_x^1\rangle\langle\mathbb{L}_2(t),h_y^2\rangle\right] \\
&=\E\left[\langle(\mathbb{L}_1(t),\mathbb{L}_2(t)),\Pi_1^*h_x^1\rangle\langle
(\mathbb{L}_1(t),\mathbb{L}_2(t)),\Pi_2^*h_y^2\rangle\right] \\
&=t\langle\mathcal{Q}\Pi_1^*h_x^1,\Pi_2^*h_y^2\rangle \\
&=t\langle\Pi_2\mathcal{Q}\Pi_1^*h_x^1,h_y^2\rangle\,.
\end{align*}
We have $\Pi_2\mathcal{Q}\Pi^*_1=\mathcal{Q}_{12}$, and it follows
\begin{equation}
\text{Cov}(\mathbb{L}_1(t,x),\mathbb{L}_2(t,y))=t\langle\mathcal{Q}_{12}h_x^1,h_y^2\rangle\,,
\end{equation}
as claimed.

Let us analyse a very simple case of the bivariate forward dynamics in \eqref{eq:bivariate-hjm-dyn} where
$\alpha_1=\alpha_2=\alpha$ and $\sigma_i=\text{Id}$, the identity operator on $H_{\alpha}$, $i=1,2$,
and $(\mathbb{L}_1,\mathbb{L}_2)=(\mathbb{B}_1,\mathbb{B}_2)$ is a Wiener process. 
The mild solution in \eqref{eq:bivariate-hjm-mild} takes the form
\begin{align*}
g_i(t)&=\mathcal{S}_tg_i^0+\int_0^t\mathcal{S}_{t-s}\,d\mathbb{B}_i(s)\,,
\end{align*}
for $i=1,2$. We find similar to above that, for $x,y\in\R_+$,
\begin{align*}
\text{Cov}(g_1(t,x),g_2(t,y))&=\E\left[\delta_x\int_0^t\mathcal{S}_{t-s}\,d\mathbb{B}_1(s)\cdot\delta_y
\int_0^t\mathcal{S}_{t-s}\mathbb{B}_2(s)\right]\\
&=\E\left[\left\langle\left(\int_0^t\mathcal{S}_{t-s}\,d\mathbb{B}_1(s),\int_0^t\mathcal{S}_{t-s}\,d\mathbb{B}_2(s)\right),\Pi_1^*h_x\right\rangle \right. \\
&\qquad\qquad\left.\times\left\langle\left(\int_0^t\mathcal{S}_{t-s}\,d\mathbb{B}_1(s),\int_0^t\mathcal{S}_{t-s}\,d\mathbb{B}_2(s)\right),\Pi_2^*h_y\right\rangle\right]\,.
\end{align*}
We show that $(\int_0^t\mathcal{S}_{t-s}\,d\mathbb{B}_1(s),\int_0^t\mathcal{S}_{t-s}\,d\mathbb{B}(s))$ is
a Gaussian $H_{\alpha}\times H_{\alpha}$-valued stochastic process:
\begin{lemma}
\label{lemma:bivariate-gaussian}
Suppose that $H_i=H_{\alpha}$ for $i=1,2$. 
The process $t\mapsto(\int_0^t\mathcal{S}_{t-s}\,d\mathbb{B}_1(s),\int_0^t\mathcal{S}_{t-s}\,d\mathbb{B}_2(s))$ is a mean-zero Gaussian $H_{\alpha}\times H_{\alpha}$-valued process
with covariance operator $\mathcal{Q}_t$ for each $t\geq 0$ given by
$$
\mathcal{Q}_t=\left[\begin{array}{cc} \int_0^t\mathcal{S}_s\mathcal{Q}_1\mathcal{S}^*_s\,ds & \int_0^t\mathcal{S}_s\mathcal{Q}_{12}^*\mathcal{S}_s\,ds \\
\int_0^t\mathcal{S}_s\mathcal{Q}_{12}\mathcal{S}_s^*\,ds & \int_0^t\mathcal{S}_s\mathcal{Q}_2\mathcal{S}_s^*\,ds\end{array}\right]\,
$$ 
The integrals in $\mathcal{Q}_t$ are interpreted as Bochner integrals in the space of Hilbert Schmidt operators.
\end{lemma}
\begin{proof}
First, note that all the integrals in $\mathcal{Q}_t$ are well-defined as Bochner integrals because
the operator norm of the involved operators are bounded uniformly in time by 
Lemma~\ref{lem:cont-shift-op}.

Consider the characteristic function of the process at time $t\geq 0$. A straightforward computation gives,
\begin{align*}
\E&\left[\exp\left(\mathrm{i}\left\langle\left(\int_0^t\mathcal{S}_{t-s}\,d\mathbb{B}_1(s),\int_0^t\mathcal{S}_{t-s}\,d\mathbb{B}(s)\right),(u,v)\right\rangle\right)\right] \\
&\qquad=\exp\left(-\frac12\int_0^t\langle
\mathcal{Q}(\mathcal{S}^*_{t-s}u,\mathcal{S}_{t-s}^*v),(\mathcal{S}_{t-s}^*u,\mathcal{S}_{t-s}^*v)\rangle\,ds\right)\,.
\end{align*}
Using the definition of $\mathcal{Q}$ shows that 
$$
\E\left[\exp\left(\mathrm{i}\left\langle\left(\int_0^t\mathcal{S}_{t-s}\,d\mathbb{B}_1(s),\int_0^t\mathcal{S}_{t-s}\,d\mathbb{B}(s)\right),(u,v)\right\rangle\right)\right] =\exp\left(-\frac12\langle\mathcal{Q}_t(u,v),(u,v)\rangle\right)\,,
$$
and the result follows.
\end{proof}
It follows from this Lemma that
\begin{align*}
\text{Cov}(g_1(t,x),g_2(t,y))&=\langle \mathcal{Q}_t\Pi_1^*h_x,\Pi_2^*h_y\rangle \\
&=\langle\Pi_2\mathcal{Q}_t\Pi_1^*h_x,h_y\rangle \\
&=\left\langle\int_0^t\mathcal{S}_s\mathcal{Q}_{12}\mathcal{S}_sh_x\,ds,h_y\right\rangle \\
&=\int_0^t\langle\mathcal{S}_s\mathcal{Q}_{12}\mathcal{S}_s^*h_x,h_y\rangle\,ds \\
&=\int_0^t\delta_y\mathcal{S}_s\mathcal{Q}_{12}\mathcal{S}_s^*\delta^*_x(1)\,ds \\
&=\int_0^t\delta_{y+s}\mathcal{Q}_{12}\delta^*_{x+s}(1)\,ds \,.
\end{align*}
This provides us with an "explicit" expression for the covariance between the forward prices 
$g_1(t)$ and $g_2(t)$ at two different maturities $x$ and $y$. 

An application of the above considerations is the pricing of so-called {\it energy quanto options}. 
Such options have gained some attention in recent years since they offer a hedge against both 
price and volume risk in energy production. A typical payoff function at exercise time $\tau$ from a quanto option takes the form
$$
p(F_{\text{energy}}(\tau,T_1,T_2))\times q(F_{\text{temp}}(\tau,T_1,T_2))\,,
$$  
where $F_{\text{energy}}$ is the forward price on some energy like power or gas, and 
$F_{\text{temp}}$ the forward price on some temperature index. Both forwards have a
delivery\footnote{Obviously, temperature is not {\it delivered}, but the temperature futures is settled 
against the measured temperature index over this period.} 
period $[T_1,T_2]$, and it is assumed $\tau\leq T_1$. The functions $p$ and $q$ are real-valued and
of linear growth, and typically given by call and put option payoff functions. Temperature is 
closely linked to the demand for power, and the quanto options are structured to yield a payoff which 
depends on the product of price and volume. We refer to Caporin, Pres and Torro \cite{CPT} and
Benth, Lange and Myklebust~\cite{BLM} for
a detailed discussion of energy quanto options. From the considerations in 
Section~\ref{sect:hilbert-space-realization}, we can
express the price at  $t\leq\tau$ of the quanto options as
\begin{equation}
V(t,g_1(t),g_2(t))=\e^{-r(\tau-t)}\E\left[p(\mathcal{L}_{\text{energy}}(g_1(\tau))
q(\mathcal{L}_{\text{temp}}(g_2(\tau))\,|\,g_1(t),g_2(t)\right]\,.
\end{equation}
Here, we have assumed that
\begin{align}
F_{\text{energy}}(t,T_1,T_2)&:=\mathcal{L}_{\text{energy}}(g_1(t)):=\delta_{T_1-t}\circ\mathcal{D}_\ell^{w,1}(g_1(t)) \\
F_{\text{temp}}(t,T_1,T_2)&:=\mathcal{L}_{\text{temp}}(g_1(t)):=\delta_{T_1-t}\circ\mathcal{D}_\ell^{w,2}(g_2(t))\,,
\end{align} 
with $\mathcal{D}_{\ell}^{w,i}$ defined as in \eqref{def-D_ell-op} using the obvious meaning of the indexing by $i=1,2$. Since $\mathcal{L}_{\text{energy}}$ and $\mathcal{L}_{\text{temp}}$ are linear
functionals on $H_{\alpha}$, it follows from Thm.~2.1 in Benth and Kr\"uhner~\cite{BK-HJM} that 
$$
(F_{\text{energy}}(t,T_1,T_2),F_{\text{temp}}(t,T_1,T_2))
$$ 
is a bivariate Gaussian random variable
on $\R^2$. From Lemma~\ref{lemma:bivariate-gaussian}, we can compute the covariance as
\begin{align*}
\text{Cov}(F_{\text{energy}}(t,T_1,T_2),F_{\text{temp}}(t,T_1,T_2))&=\E\left[\mathcal{L}_{\text{energy}}
\int_0^t\mathcal{S}_{t-s}\,d\mathbb{B}_1(s)\cdot\mathcal{L}_{\text{temp}}\int_0^t\mathcal{S}_{t-s}\,d\mathbb{B}_2(s)\right] \\
&=\E\left[\langle(\int_0^t\mathcal{S}_{t-s}\,d\mathbb{B}_1(s),\int_0^t\mathcal{S}_{t-s}\,d\mathbb{B}_2(s)),\Pi_1^*\mathcal{L}_{\text{energy}}^*(1)\rangle \right. \\
&\qquad\left.\times\langle(\int_0^t\mathcal{S}_{t-s}\,d\mathbb{B}_1(s),\int_0^t\mathcal{S}_{t-s}\,d\mathbb{B}_2(s)),\Pi^*_2\mathcal{L}^*_{\text{temp}}(1)\rangle\right] \\
&=\langle\mathcal{Q}_t(\Pi^*_1\mathcal{L}_{\text{energy}}^*(1),\Pi_2^*\mathcal{L}_{\text{temp}}^*(1)),(\Pi^*_1\mathcal{L}_{\text{energy}}^*(1),\Pi_2^*\mathcal{L}_{\text{temp}}^*(1))\rangle \\
&=\int_0^tC_{12}(s)\,ds\,,
\end{align*}
where
\begin{align*}
C_{12}(s)&=\mathcal{L}_{\text{energy}}\Pi_1\mathcal{S}_s\mathcal{Q}_1\mathcal{S}_s^*
\Pi^*_1\mathcal{L}^*_{\text{energy}}(1)+\mathcal{L}_{\text{energy}}\Pi_1\mathcal{S}_s\mathcal{Q}_{12}^*
\mathcal{S}_s^*\Pi_2^*\mathcal{L}_{\text{temp}}^*(1) \\
&\qquad+\mathcal{L}_{\text{temp}}\Pi_2\mathcal{S}_s\mathcal{Q}_{12}\mathcal{S}_s^*\Pi_1^*\mathcal{L}_{\text{energy}}^*(1)
+\mathcal{L}_{\text{temp}}\Pi_2\mathcal{S}_s\mathcal{Q}_{2}\mathcal{S}_s^*\Pi_2^*\mathcal{L}_{\text{temp}}^*(1)\,.
\end{align*}
Thus, we can obtain a price $V(t,g_1(t),g_2(t))$ in terms of an integral with respect to a
Gaussian bivariate probability distribution, involving similar operators (and their duals) as for the 
European options studied in Section~\ref{sect:european-option}. We remark in passing that Benth, Lange and Myklebust~\cite{BLM}
derive a Black \& Scholes-like pricing formula for a call-call quanto options, which is applied to price
such derivatives written on Henry Hub gas futures traded at NYMEX and HDD/CDD temperature futures
traded at CME.

We next return back to the general considerations on the factorization of the covariance operator
$\mathcal{Q}$ of a bivariate square-integrable random variable in $H_1\times H_2$. 
If we want to construct an operator $\mathcal{Q}$ as in 
Thm~\ref{s:Notwendigkeit A}, then the operator $\mathcal{Q}_{12}$ appearing there 
has necessarily has to satisfy condition (ii). As we will show in the next theorem, condition (ii) of Thm~\ref{s:Notwendigkeit A} is sufficient as well.
\begin{theorem}\label{s:Hinreichend A}
 Let $H_i$ be a separable Hilbert space, $\mathcal{Q}_i$ be a positive semidefinite 
trace class operator on $H_i$ and define 
$\mathcal{R}_i:=\sqrt{\mathcal{Q}_i}$ for $i=1,2$. 
Let $\mathcal{Q}_{12}\in L(H_1,H_2)$ such that
 \begin{align*} \vert \< \mathcal Q_{12}u,v\>\vert \leq \Vert R_1u\Vert_1\Vert R_2v\Vert_2 \end{align*}
 for any $u\in H_1$, $v\in H_2$. Then, 
$$
\mathcal{Q}:=\left(\begin{matrix}\mathcal{Q}_1&\mathcal{Q}_{12}^*\\
\mathcal{Q}_{12}&\mathcal{Q}_2\end{matrix}\right)\,,
$$
defines a positive semidefinite operator on $H_1\times H_2$. Moreover, 
$\mathcal{Q}$ is positive definite if and only if $\mathcal{Q}_1$, 
$\mathcal{Q}_2$ are positive definite and  
 \begin{align*} \vert \< \mathcal Q_{12}u,v\>\vert < \Vert \mathcal R_1u\Vert_1\Vert \mathcal  R_2v\Vert_2 \end{align*}
for any $u\in H_1\backslash\{0\}$, $v\in H_2\backslash\{0\}$.
\end{theorem}
\begin{proof}
 Let $u\in H_1$ and $v\in H_2$. Then
 \begin{align*}
   \<\mathcal{Q}(u,v),(u,v)\> &= \<\mathcal{Q}_1u,u\> + \<\mathcal{Q}_2v,v\> + 2\<\mathcal{Q}_{12}u,v\> \\
                            &\geq \Vert \mathcal R_1u\Vert_1^2 + \Vert \mathcal R_2v\Vert_2^2 - 2\Vert \mathcal R_1u\Vert_1\Vert \mathcal  R_2v\Vert_2 \\
                            &= (\Vert \mathcal R_1u\Vert_1 - \Vert \mathcal  R_2v\Vert_2)^2 \\
                            &\geq 0.
 \end{align*}
Under the additional assumptions, the first inequality is strict.
\end{proof}

We now analyse the pricing of spread options in a simple setting:
Let us consider a "bivariate" exponential model $g_i(t)=\exp(\widetilde{g}_i(t))$, $i=1,2$, defined on the space 
$H_{\alpha}\times H_{\alpha}$ by a dynamics similar to \eqref{eq:bivariate-hjm-dyn} (but with a drift) driven by 
$(\mathbb{L}_1(t),\mathbb{L}_2(t))=(\mathbb{B}_1(t),\mathbb{B}_2(t))$:
 \begin{align*}
d\widetilde{g}_1(t)&=\partial_x \widetilde{g}_1(t)\,dt+\mu_1(t)\,dt+\sigma_1(t)\,d\mathbb{B}_1(t) \\
d\widetilde{g}_2(t)&=\partial_x \widetilde{g}_2(t)\,dt+\mu_2(t)\,dt+\sigma_2(t)\,d\mathbb{B}_2(t)\,.
\end{align*}
Here, we suppose that $\sigma_i:\R_+\rightarrow L(H_{\alpha})$ is non-random and 
$\sigma_i\in\mathcal{L}_{\mathbb{B}_i}^2(H_{\alpha})$, $i=1,2$. 
Thus, we have the forward price dynamics $f_i(\tau,T)$ given $f_i(t,T)$ for $t\leq\tau\leq T$,
\begin{equation}
f_i(\tau,T)=f_i(t,T)\exp\left(\int_t^{\tau}\delta_{T-s}\mu_i(s)\,ds+\delta_{T-\tau}\int_t^{\tau}\mathcal{S}_{\tau-s}\sigma_i(s)\,d\mathbb{B}_i(s)\right)\,,
\end{equation}
for $i=1,2$. Introduce the notation
\begin{equation}
\widetilde{\sigma}_i^2(s,T)=\delta_{T-s}\sigma_i(s)\mathcal{Q}_i\sigma_i^*(s)\delta_{T-s}^*(1)\,.
\end{equation}
for $i=1,2$. From Thm.~2.1 in Benth and Kr\"uhner~\cite{BK-HJM} it follows for $i=1,2$,
$$
\delta_{T-\tau}\int_t^{\tau}\mathcal{S}_{\tau-s}\sigma_i(s)\,d\mathbb{B}_i(s)=\int_t^{\tau}\widetilde{\sigma}_i(s,T)\,dB_i(s)\,,
$$
where $B_i$ is a real-valued Brownian motion. 
By Lemma~\ref{lemma:drift-condition}, we have the no-arbitrage drift condition
\begin{equation}
\widetilde{\mu}_i(s,T):=\delta_{T-s}\mu_i(s)=-\frac12\widetilde{\sigma}^2_i(s,T)\,.
\end{equation}
Remark that, as a consequence of the non-random assumption on $\sigma_i(s)$,  
$\int_t^{\tau}\widetilde{\sigma}_i(s,T)\,dB_i(s), i=1,2$ are two Gaussian random variables on $\R$ with mean zero and
variance $\int_t^{\tau}\widetilde{\sigma}_i^2(s,T)\,ds, i=1,2$, resp. Moreover, a direct computation using the above
theory reveals the covariance between these two random variables:
\begin{align*}
\E&\left[\int_t^{\tau}\widetilde{\sigma}_1(s,T)\,dB_1(s)\int_t^{\tau}\widetilde{\sigma}_2(s,T)\,dB_2(s)\right] \\
&\qquad=\E\left[\delta_{T-\tau}\int_t^{\tau}\mathcal{S}_{\tau-s}\sigma_1(s)\,d\mathbb{B}_1(s)
\times\delta_{T-\tau}\int_t^{\tau}\mathcal{S}_{\tau-s}\sigma_2(s)\,d\mathbb{B}_2(s)\right] \\
&\qquad=\E\left[\langle\int_t^{\tau}\mathcal{S}_{\tau-s}\sigma_1(s)\,d\mathbb{B}_1(s),h_{T-\tau}\rangle
\langle\int_t^{\tau}\mathcal{S}_{\tau-s}\sigma_2(s)\,d\mathbb{B}_2(s),h_{T-\tau}\rangle\right] \\
&\qquad=\int_t^{\tau}\langle\mathcal{Q}\Pi_1^*\sigma_1^*(s)\mathcal{S}_{\tau-s}^*h_{T-\tau},\Pi_2^*\sigma_2^*(s)
\mathcal{S}_{\tau-s}^*h_{T-\tau}\rangle\,ds \\
&\qquad=\int_t^{\tau}\delta_{T-s}\sigma_2(s)\mathcal{Q}_{12}\sigma^*_1(s)\delta_{T-s}^*(1)\,ds \\
&\qquad:=\int_t^{\tau}\widetilde{\sigma}_{12}(s,T)\,ds\,.
\end{align*}
Hence, for $i=1,2$,
\begin{equation}
\label{eq:bivariate-exp-dyn}
f_i(\tau,T)=f_i(t,T)\exp\left(-\frac12\int_t^{\tau}\widetilde{\sigma}_i^2(s,T)\,ds+\int_t^{\tau}\widetilde{\sigma}_i(s)\,dB_i(s)\right)\,,
\end{equation}
where we know that the two stochastic integrals form a bivariate Gaussian random variable with known variance-covariance matrix. 
The price at time $t$ of a call option written on the spread between the two forwards with exercise at time $t\leq \tau\leq T$ will be
$$
V(t)=\e^{-r(\tau-t)}\E\left[\max\left(f_1(\tau,T)-f_2(\tau,T),0\right)\,|\,\mathcal{F}_t\right]\,.
$$ 
Using the representation of the forward prices in \eqref{eq:bivariate-exp-dyn} we find the spread option pricing 
formula
\begin{equation}
V(t)=\e^{-r(\tau-t)}\left\{f_1(t,T)\Phi(d_+)-f_2(t,T)\Phi(d_-)\right\}\,,
\end{equation}  
where $\Phi$ is the cumulative standard normal distribution function,
$$
d_{\pm}=\frac{\ln(f_1(t,T)/f_2(t,T))\pm\Sigma^2(t,\tau,T)/2}{\Sigma(t,\tau,T)}\,,
$$
and
$$
\Sigma^2(t,\tau,T)=\int_t^{\tau}\widetilde{\sigma}_1^2-2\widetilde{\sigma}_{12}(s,T)+\widetilde{\sigma}_2^2(s,T)\,ds\,.
$$
We have recovered the Margrabe formula (see Margrabe~\cite{Margrabe}) with time-dependent volatility and correlation. 
Observe that the spread option price becomes a function of the initial forward prices at time $t$ with delivery at time $T$.

We proceed with some more general considerations on "bivariate" random variables in Hilbert spaces and their 
representation.
If $(X,Y)$ is a $2$-dimensional Gaussian random variable, we know from classical probability theory that 
there exist a Gaussian random variable $Z$ being independent of $X$ and $a\in\mathbb R$ such that $Y = aX+Z$. The next Proposition is a generalisation of this statement to square-integrable Hilbert-space valued
random variables: 
\begin{proposition}\label{p:unkorrelierte Darstellung}
 Let $X_i$ be an $H_i$-valued square-integrable random variable with covariance 
$\mathcal{Q}_i$ and let $\mathcal{Q}_{12}\in L(H_1,H_2)$ be the operator given in 
Thm~\ref{s:Notwendigkeit A} such that
$$
\mathcal{Q}:=\left(\begin{matrix}\mathcal{Q}_1&\mathcal{Q}_{12}^*\\
\mathcal{Q}_{12}&\mathcal{Q}_2\end{matrix}\right)\,,
$$
is the covariance operator of $(X_1,X_2)$. Assume that $\mathrm{ran}(\mathcal Q^*_{12})\subseteq \mathrm{ran}(\mathcal Q_1)$. Then the closure $\mathcal B$ of the densely defined operator $\mathcal Q_{12}\mathcal Q_1^{-1}$ is in $L(H_1,H_2)$ where $\mathcal Q_1^{-1}$ denotes the pseudo-inverse of $\mathcal Q_1$. Define $Z:= X_2 - \mathcal BX_1$. Then $Z$ is a centered, square integrable and $H_2$-valued random variable with $\E(\<X_1,u\>\<Z,v\>)=0$ for any $u\in H_1$, $v\in H_2$, i.e.\ $X_1$ and $Z$ are uncorrelated.

In particular, the covariance operator of $(X_1,Z)$ is given by
$$
\mathcal{Q}_{X_1,Z}:=\left(\begin{matrix}\mathcal{Q}_1&0\\0&\mathcal{Q}_Z
\end{matrix}\right)\,,
$$
where $\mathcal Q_Z$ denotes the covariance operator of $Z$.
\end{proposition}
\begin{proof}
  $\mathcal Q_{12}\mathcal Q_1^{-1}$ is densely defined because its domain is the domain of $\mathcal Q_1^{-1}$. Define $\mathcal C := \mathcal Q_1^{-1}\mathcal Q_{12}^*$ which is a closed operator whose domain is $H_2$ by assumption. The closed graph theorem yields that $\mathcal C$ is continuous and linear. Consequently, its dual is a continuous linear continuation of $\mathcal Q_{12}\mathcal Q_1^{-1}$. However, the latter operator is densely defined and hence $\mathcal B:=\mathcal C^*$ is its closure. Now, let $u\in H_1$, $v\in H_2$. Then
   \begin{align*}
     \E[\<X_1,u\>\<\mathcal BX_1,v\>] &= \<\mathcal Q_1u,\mathcal B^*v\> \\
                                      &= \<\mathcal Q_1u,\mathcal Q_1^{-1}\mathcal Q_{12}^*v\> \\
                                      &= \<\mathcal Q_{12}u,v\> \\
                                      &= \E[\<X_1,u\>\<X_2,v\>].
   \end{align*}
  Thus $X_1$ and $Z$ are uncorrelated and the claim follows.
\end{proof}

This result can be applied to state a representation of the $H_1\times H_2$-valued Wiener process 
$(\mathbb{B}_1,\mathbb{B}_2)$.
\begin{proposition}\label{prop:bivariate-brownian-repr}
Let $\mathbb B_1$, $\mathbb B_2$ be $H_1$ resp.\ $H_2$ valued Brownian motions where $H_1$, $H_2$ are separable Hilbert spaces. Suppose that the random variables $\mathbb{B}_i(1)$, $i=1,2$ satisfy the conditions in Proposition \ref{p:unkorrelierte Darstellung}. Then, there exists an operator $\mathcal{B}\in L(H_1,H_2)$ such that $\mathbb W:=\mathbb B_2 - \mathcal B\mathbb B_1$ is an $H_2$-valued Brownian motion which is independent of $H_1$.
\end{proposition}
\begin{proof}
 Let $\mathcal B$ be the operator given in Proposition \ref{p:unkorrelierte Darstellung} for the random variables $\mathbb B_1$ and $\mathbb B_2$. Then, $(\mathbb B_1,\mathbb W)$ is an other Brownian motion. Moreover,
  $$ \E[\<\mathbb B_1(t),u\>\<\mathbb W(t),v\>] = t \E[\<\mathbb B_1(1),u\>\<\mathbb W(1),v\>] = 0 $$
 for any $t\geq0$. The claim follows.
\end{proof}
The proposition allows us to model a "bivariate" forward dynamics driven by two dependent
Brownian motions
\begin{align*}
dg_1(t)&=\partial_x g_1(t)\,dt+\sigma_1(t,g_1(t),g_2(t))\,d\mathbb{B}_{1}(t) \\
dg_2(t)&=\partial_x g_2(t)\,dt+\sigma_2(t,g_1(t),g_2(t))\,d\mathbb{B}_{2}(t) \,,
\end{align*}
by a dynamics driven by two independent Brownian motions,
\begin{align*}
dg_1(t)&=\partial_x g_1(t)\,dt+\sigma_1(t,g_1(t),g_2(t))\,d\mathbb{B}_{1}(t) \\
dg_2(t)&=\partial_x g_2(t)\,dt+
\sigma_2(t,g_1(t),g_2(t))\,d\mathbb{W}(t)-\sigma_2(t,g_1(t),g_2(t))\mathcal{B}\,d\mathbb{B}_1(t) \,.
\end{align*}
Here, the operator $\mathcal{B}$ plays the role of a "correlation" coefficient, describing
how the two noises $\mathbb{B}_1$ and $\mathbb{B}_2$ depend statistically. Indeed, choosing
$H_i=H_{\alpha}$, $i=1,2$ to be the Filipovic space, we see that 
\begin{align*}
\E[\delta_x\mathbb{B}_1(t)\delta_y\mathbb{B}_2(t)]
&=\E[\<\mathbb{B}_1(t),h_x\>\<\mathbb{B}_2(t),h_y\>] \\
&=\E[\<\mathbb{B}_1(t),h_x\>\<\mathcal{B}\mathbb{B}_1(t),h_y\>]\\
&=t\<\mathcal{B}\mathcal{Q}_1h_x,h_y\> \\
&=t\delta_y\mathcal{B}\mathcal{Q}_1\delta_x^*(1)\,,
\end{align*}
for $x,y\in\R_+$. Hence, the correlation between $\mathbb{B}_1(t,x)$ and $\mathbb{B}_2(t,y)$ is
modelled by the operator $\mathcal{B}$. We can derive a similar representation for two L\'evy processes, but
they will not be independent but only uncorrelated in most cases.

As a final remark we like to note that the 'odd' range condition in Proposition \ref{prop:bivariate-brownian-repr} is needed to ensure the existence of a linear operator from $H_1$ to $H_2$. However, in the Gaussian case it is possible to find a linear operator $\mathcal T$ from $L^2(\Omega,\mathcal A,P,H_1)$ to $L^2(\Omega,\mathcal A,P,H_2)$ yielding an independent decomposition of the second factor. We now give the precise statement.
\begin{proposition}
  Let $H_1$, $H_2$ be separable Hilbert spaces and $(X_1,X_2)$ be an $H_1\times H_2$ valued Gaussian random variable. Let $\mathcal B$ be the closure of $Q_{12}^*Q_1^{-1}$. Then, $P(X_1 \in \mathrm{dom}(\mathcal B)) = 1$ and
  $Z := X_2-\mathcal BX_1$ is Gaussian and $X_1,Z$ are independent.
\end{proposition}
\begin{proof}
 Let $(e_n)_{n\in\mathbb N}$ be an orthonormal basis of $H_1$ such that $X_1 = \sum_{n=1}^\infty \sqrt{\lambda_n}\Phi_ne_n$ where $(\Phi_n)_{n\in\mathbb N}$ is a sequence of i.i.d.\ standard normal random variables $\lambda_n\geq0$ and $\sum_{n\in\mathbb N}\lambda_n<\infty$, cf.\ Peszat and Zabczyk~\cite[Thm.~4.20]{peszat.zabczyk.07}. Define $Y_k:=\sum_{n=1}^k \sqrt{\lambda_n}\Phi_ne_n$ for any $k\in\mathbb N$. Clearly, we have $Y_k\rightarrow X_1$ for $k\rightarrow\infty$. We now want to show that $\mathcal BY_k$ converges to $\E[X_2|X_1]$ which will complete the proof.
 
 Let $(p_j)_{j\in\mathbb N}$ be the hermite polynomials on $\mathbb R$. Then $\E[p_j(\Phi_1)p_i(\Phi_1)] = 1_{\{i=j\}}$ for any $i,j\in\mathbb N$. For an $H_2$-valued square integrable random variable $A$ we have
   $$ \E[A\vert X_1] = \sum_{n,m,j=1}^\infty \E[\<A,f_m\>p_j(\Phi_n)]p_j(\Phi_n)f_m $$
 where $(f_m)_{m\in\mathbb N}$ is an orthonormal basis of $H_2$. Since $(X_1,X_2)$ is Gaussian, $(\Phi_n,\<X_2,f_m\>)$ is Gaussian for any $n,m\in \mathbb N$. Thus, we have
   \begin{align*}
     \E[X_2\vert X_1] &= \sum_{n,m,j=1}^\infty \E[\<X_2,f_m\>p_j(\Phi_n)]p_j(\Phi_n)f_m \\
                    &= \sum_{n,m=1}^\infty \E[\<X_2,f_m\>\Phi_n]\Phi_nf_m
   \end{align*}
 because $\E[Ap_j(B)]=0$ whenever $(A,B)$ is a normal random variable in $\mathbb R^2$, $B$ is standard normal and $j\neq 1$. Moreover, $\Phi_n = \frac{\<X_1,e_n\>}{\sqrt{\lambda_n}}$ and hence
 \begin{align*}
   \E[\<X_2,f_m\>\Phi_n] &= \frac{\<\mathcal Q_{12}e_n,f_m\>}{\sqrt{\lambda_n}}, \\
    \Phi_n &= \sqrt{\lambda_n}  \<X_1,\mathcal Q_1^{-1}e_n\>, \\
   \E[X_2\vert X_1] &= \sum_{n,m=1}^\infty \<\mathcal Q_{12}e_n,f_m\> \<X_1,\mathcal Q_1^{-1}e_n\>f_m \\
                    &= \sum_{n=1}^\infty \<X_1,\mathcal Q_1^{-1}e_n\>\mathcal Q_{12}e_n
 \end{align*}
  Thus, we have
  \begin{align*}
     \mathcal BY_k &= \sum_{n=1}^k \<Y_k,\mathcal Q_1^{-1}e_n\>\mathcal Q_{12}e_n \\
                      &= \sum_{n=1}^k \<X_1,\mathcal Q_1^{-1}e_n\>\mathcal Q_{12}e_n \\
                      &\rightarrow \E[X_2|X_1]
  \end{align*}
  for $k\rightarrow \infty$ where we used Parseval's identity for the first equality. Since $\mathcal B$ is closed we have $X_1 \in \mathrm{dom}(\mathcal B)$ $P$-a.s.\ and $\mathcal BX_1 = \E[X_2|X_1]$. 
\end{proof}

\end{document}